\documentclass[sigconf]{acmart}

\usepackage{subcaption}
\usepackage{stmaryrd}
\usepackage{amsmath}
\usepackage{import}
\usepackage{float}
\usepackage{graphicx}
\usepackage{epstopdf}
\usepackage{amsfonts}
\usepackage[ruled,vlined]{algorithm2e} 

\SetCommentSty{mycommfont}

\usepackage{multirow}
\usepackage{microtype}
\usepackage{balance}
\usepackage{tikz}
\usetikzlibrary{plotmarks}
\usepackage{pgfplots}

\usepackage{amsthm}
\usepackage{pgfplotstable}
\pgfplotsset{compat=newest}
\usepackage{hyperref}
\usepackage{color}
\usepackage{enumitem}
\usepackage{colortbl}
\usepackage{wrapfig,lipsum,booktabs}
\usepackage{cuted}
\usepackage{setspace}

\theoremstyle{plain}
\newtheorem{thm}{Theorem}[section]

\theoremstyle{definition}
\newtheorem{defn}{Definition}[section]

\usepackage{xspace}
\newcommand{\cf}{\hbox{\emph{cf.}}\xspace}

\newcommand{\eg}{\hbox{\emph{e.g.}}\xspace}
\newcommand{\ie}{\hbox{\emph{i.e.}}\xspace}

\newcommand{\wrt}{\hbox{\emph{w.r.t.}}\xspace}

\definecolor{light-gray}{gray}{0.7}
\newcommand{\highlight}[1]{\colorbox{light-gray}{#1}}
\newcommand{\mpise}{MPI-SV\xspace}
\newcommand{\lazy}{blocking-driven\xspace}
\newcommand{\LAZY}{Blocking-driven\xspace}
\newcommand{\issue}[1]{\mathit{issue}({#1})}
\newcommand{\enable}[1]{\mathit{enabled}({#1})}

\usepackage{booktabs}   
\usepackage{subcaption} 

\begin{document}
\copyrightyear{2020}
\acmYear{2020}
\setcopyright{acmcopyright}
\acmConference[ICSE '20]{ICSE '20: 42nd International Conference on Software Engineering }{May 23-29, 2018}{Seoul, South Korea}
\acmBooktitle{ICSE '20: ICSE '20: 42nd International Conference on Software Engineering , May 23-29, 2020, Seoul, South Korea}
\acmPrice{15.00}
\acmDOI{XX.XXXX/XXXXXXX.XXXXXXX}
\acmISBN{xxx-x-xxxx-xxxx-x/xx/xx}
\def\supplementary{supplementary}

\title{Combining Symbolic Execution and Model Checking to Verify MPI Programs}
\author{Hengbiao Yu\textsuperscript{1$*$}, Zhenbang Chen\textsuperscript{1$*$\authornote{The first two authors contributed equally to this work and are co-first authors. Zhenbang Chen and Ji Wang are the corresponding authors.}}, Xianjin Fu\textsuperscript{1,2}, Ji Wang\textsuperscript{1,2$*$}, Zhendong Su\textsuperscript{3},}
\author{
Jun Sun\textsuperscript{4}, Chun Huang\textsuperscript{1}, Wei Dong\textsuperscript{1}
}
\renewcommand{\authors}{Hengbiao Yu, Zhenbang Chen, Xianjin Fu, Ji Wang, Zhendong Su, Jun Sun, Chun Huang, and Wei Dong}
\renewcommand{\author}{Hengbiao Yu, Zhenbang Chen, Xianjin Fu, Ji Wang, Zhendong Su, Jun Sun, Chun Huang, and Wei Dong}
\affiliation{\textsuperscript{1}College of Computer, National University of Defense Technology, Changsha, China} 
\affiliation{\textsuperscript{2}State Key Laboratory of High Performance Computing, National University of Defense Technology, Changsha, China}
\affiliation{\textsuperscript{3}Department of Computer Science, ETH Zurich, Switzerland}
\affiliation{\textsuperscript{4}School of Information Systems, Singapore Management University, Singapore}
\email{{hengbiaoyu, zbchen, wj}@nudt.edu.cn, zhendong.su@inf.ethz.ch, junsun@smu.edu.sg, wdong@nudt.edu.cn}



\begin{abstract}
Message passing is the standard paradigm of programming in high-performance computing. However, verifying Message Passing Interface (MPI) programs is challenging, due to the complex program features (such as non-determinism and non-blocking operations).
In this work, we present MPI symbolic verifier (MPI-SV), the first symbolic execution based tool for automatically verifying MPI programs with non-blocking operations. MPI-SV combines symbolic execution and model checking in a synergistic way to tackle the challenges in MPI program verification. The synergy improves the scalability and enlarges the scope of verifiable properties.
We have implemented MPI-SV\footnote{MPI-SV is available  \url{https://mpi-sv.github.io}.} and evaluated it with 111 real-world MPI verification tasks. The pure symbolic execution-based technique successfully verifies 61 out of the 111 tasks (55\%) within one hour, while in comparison, MPI-SV verifies 100 tasks (90\%). On average, compared with pure symbolic execution, MPI-SV achieves 19x speedups on verifying the satisfaction of the critical property and 5x speedups on finding violations.
\end{abstract}






\maketitle
\renewcommand{\shortauthors}{Hengbiao Yu, Zhenbang Chen, Xianjin Fu, Ji Wang, Zhendong Su, Jun Sun, Chun Huang, and Wei Dong}

\section{Introduction}

Nowadays, an increasing number of high-performance computing (HPC) applications have been developed to solve large-scale problems~\citep{buyya1999high}. The Message Passing Interface (MPI)~\citep{snir1998mpi} is the current \emph{de facto} standard programming paradigm for developing HPC applications. Many MPI programs are developed with significant human effort. One of the reasons is that MPI programs are \emph{error-prone} because of complex program features (such as \emph{non-determinism} and \emph{asynchrony}) and their scale. Improving the reliability of MPI programs is challenging ~\citep{DBLP:journals/cacm/GopalakrishnanKSTGLSSB11,HPCSummit}.

Program analysis~\citep{nielson2015principles} is an effective technique for improving program reliability. Existing methods for analyzing MPI programs can be categorized into \emph{dynamic} and \emph{static} approaches. Most existing methods are dynamic, such as debugging~\citep{DBLP:journals/cacm/LagunaASGLSBKZC15}, correctness checking~\citep{DBLP:conf/parco/SamofalovKKZKD05} and dynamic verification~\citep{DBLP:conf/cav/VakkalankaGK08}. These methods need concrete inputs to run MPI programs and perform analysis based on runtime information. Hence, dynamic approaches may miss input-related program errors. Static approaches~\citep{DBLP:conf/pvm/Siegel07,DBLP:conf/oopsla/LopezMMNSVY15,DBLP:conf/cgo/Bronevetsky09,DBLP:conf/vmcai/BotbolCG17} analyze abstract models of MPI programs and suffer from false alarms, manual effort, and poor scalability. To the best of our knowledge, existing \emph{automated verification} approaches for MPI programs either do not support \emph{input-related} analysis or fail to support the analysis of the MPI programs with 
\emph{non-blocking} operations, the invocations of which do not block the program execution. Non-blocking operations are ubiquitous in real-world MPI programs for improving the performance but introduce more complexity to programming. 

Symbolic execution~\citep{king1976symbolic,godefroid2005dart} supports input-related analysis by systematically exploring a program's path space. In principle, symbolic execution provides a balance between concrete execution and static abstraction with improved input coverage or more precise program abstraction. However, symbolic execution based analyses suffer from path explosion due to the exponential increase of program paths \wrt the number of conditional statements. The problem is particularly severe when analyzing MPI programs because of parallel execution and non-deterministic operations. Existing symbolic execution based verification approaches \citep{DBLP:journals/mics/SiegelZ11}\citep{DBLP:conf/hase/FuCZHDW15} do not support non-blocking MPI operations. 

In this work, we present \mpise, a novel verifier for MPI programs by smartly integrating symbolic execution and model checking. As far as we know, \mpise\ is the first automated verifier that supports non-blocking MPI programs and LTL \citep{DBLP:books/daglib/0077033} property verification. \mpise uses symbolic execution to extract \emph{path-level} models from MPI programs and verifies the models \wrt the expected properties by model checking~\citep{clarke1999model}. The two techniques complement each other: (1) symbolic execution abstracts the control and data dependencies to generate verifiable models for model checking, and (2) model checking improves the scalability of symbolic execution by leveraging the verification results to prune redundant paths and enlarges the scope of verifiable properties of symbolic execution.

In particular, MPI-SV combines two algorithms: (1) symbolic execution of \emph{non-blocking} MPI programs with \emph{non-deterministic} operations, and (2) modeling and checking the behaviors of an MPI program path precisely. 
To safely handle non-deterministic operations, the first algorithm delays the message matchings of non-deterministic operations as much as possible. 
The second algorithm extracts a model from an MPI program path. 
The model represents all the path's equivalent behaviors, \ie, the paths generated by changing the interleavings and matchings of the communication operations in the path. We have proved that our modeling algorithm is precise and consistent with the MPI standard~\cite{MPI}. We feed the generated models from the second algorithm into a model checker to perform verification \wrt the expected properties, \ie, \emph{safety} and \emph{liveness} properties in linear temporal logic (LTL) \citep{DBLP:books/daglib/0077033}. If the extracted model from a path $p$ satisfies the property $\varphi$, $p$'s equivalent paths can be safely pruned; otherwise, if the model checker reports a counterexample, a violation of $\varphi$ is found. This way, we significantly boost the performance of symbolic execution by pruning a large set of paths which are equivalent to certain paths that have been already model-checked.

We have implemented \mpise\ for MPI C programs based on Cloud9~\citep{bucur2011parallel} and PAT~\citep{sun2009pat}. We have used \mpise\ to analyze 12 real-world MPI programs, totaling 47K lines of code (LOC) (three are beyond the scale that the state-of-the-art MPI verification tools can handle), \wrt the deadlock freedom property and \emph{non-reachability} properties. For the 111 deadlock freedom verification tasks, when we set the time threshold to be an hour, \mpise\ can complete 100 tasks, \emph{i.e.}, deadlock reported or deadlock freedom verified, while pure symbolic execution can complete 61 tasks. For the 100 completed tasks, \mpise\ achieves, on average, 19x speedups on verifying deadlock freedom and 5x speedups on finding a deadlock. 

The main contributions of this work are:
\begin{itemize}[leftmargin=1em]
\setlength{\itemsep}{0pt}\setlength{\parsep}{0pt}\setlength{\parskip}{0pt}
\item A synergistic framework combining symbolic execution and model checking for verifying MPI programs.
\item A method for symbolic execution of non-blocking MPI programs with non-deterministic operations. The method is formally proven to preserve the correctness of verifying reachability properties.
\item A precise method for modeling the equivalent behaviors of an MPI path, 
which enlarges the scope of the verifiable properties and improves the scalability.
\item A tool for symbolic verification of MPI C programs and an extensive evaluation on real-world MPI programs.
\end{itemize}

\section{Illustration}\label{motivation}

In this section, we first introduce MPI programs and use an example to illustrate the problem that this work targets. Then, we overview \mpise\ informally by the example.

\subsection{MPI Syntax and Motivating Example}\label{sec:2-1}

MPI implementations, such as MPICH~\citep{gropp2002mpich2} and OpenMPI~\citep{gabriel2004open}, provide the programming interfaces of message passing to support the development of parallel applications. An MPI program can be implemented in different languages, such as C and C++. Without loss of generality, we focus on MPI programs written in C.
Let $\mathbb{T}$ be a set of types, $\mathbb{N}$ a set of names, and $\mathbb{E}$ a set of expressions.
For simplifying the discussion, we define a core language for MPI processes in Figure \ref{fig:mpisyntax},  where $\mathbf{T} \in \mathbb{T}$, $\verb"r" \in \mathbb{N}$, and $\verb"e" \in \mathbb{E}$.
An MPI program $\mathcal{MP}$ is defined by a \emph{finite} set of processes $\{\textsf{Proc}_i \mid 0 \le i \le n \}$. 
\emph{For brevity, we omit complex language features (such as the messages in the communication operations and pointer operations) although  \mpise does support real-world MPI C programs.}

\begin{figure}[!t]
\begin{center}{
\small
\begin{tabular}{rll}
$\textsf{Proc}$ & $ ::=$ & $\mathbf{var}\ \verb"r":\mathbf{T} \mid \verb"r" := \verb"e" \mid \textsf{Comm} \mid \textsf{Proc}\ ;\ \textsf{Proc}\mid $\\
& &$ \mathbf{if}\ \verb"e" \ \textsf{Proc}\ \mathbf{else}\ \textsf{Proc} \mid \mathbf{while}\ \verb"e"\ \mathbf{do}\ \textsf{Proc}$\\
$\textsf{Comm}$ & $::=$ & $\verb"Ssend(e)" \mid \verb"Send(e)" \mid \verb"Recv(e)" \mid \verb"Recv(*)" \mid \verb"Barrier" \mid $\\
&& $\verb"ISend(e,r)" \mid \verb"IRecv(e,r)" \mid \verb"IRecv(*,r)" \mid \verb"Wait(r)"$
\end{tabular}}
\end{center}
\caption{Syntax of a core MPI language.}\label{fig:mpisyntax}
\end{figure}

The statement $\mathbf{var}\ \verb"r":\mathbf{T}$ declares a variable $\verb"r"$ with type $\mathbf{T}$. The statement $\verb"r" := \verb"e"$ assigns the value of expression $\verb"e"$ to variable $\verb"r"$. A process can be constructed from basic statements by using the composition operations including sequence, condition and loop.
For brevity, we incorporate the key message passing operations in the syntax, where $\verb"e"$ indicates the destination process's identifier. These message passing operations can be \emph{blocking} or \emph{non-blocking}. First, we introduce blocking operations:
\begin{itemize}[leftmargin=1em]
\item
\verb"Ssend(e)": sends a message to the \emph{e}th process,
and the sending process blocks until the message is received by the destination process.
\item
\verb"Send(e)": sends a message to the \emph{e}th process,
and the sending process blocks until the message is copied into the system buffer.
\item
\verb"Recv(e)": receives a message from the \emph{e}th process, and
the receiving process blocks until the message from the \emph{e}th process is received.
\item
\verb"Recv(*)": receives a message from \emph{any} process,
and the receiving process blocks until a message is received regardless which process sends the message.
\item
\verb"Barrier": blocks the process until all the processes have called \verb"Barrier".
\item
\verb"Wait(r)": the process blocks until the operation indicated by \verb"r" is completed.
\end{itemize}

A \verb"Recv(*)" operation, called \emph{wildcard receive}, may receive a message from different processes under different runs, resulting in non-determinism. The blocking of a \verb"Send(e)" operation depends on the size of the system buffer, which may differ under different MPI implementations. 
For simplicity, we assume that the size of the system buffer is infinite. Hence, each \verb"Send(e)" operation returns \emph{immediately} after being issued. Note that our implementation allows users to configure the buffer size. 
To improve the performance, the MPI standard provides non-blocking operations to overlap computations and communications.
\begin{itemize}[leftmargin=1em]
\item
\verb"ISend(e,r)": sends a message to the \emph{e}th process, and the operation returns immediately after being issued. The parameter \verb"r" is the handle of the operation.
\item
\verb"IRecv(e,r)": receives a message from the \emph{e}th process, and the operation returns immediately after being issued. \verb"IRecv(*,r)" is the non-blocking wildcard receive.
\end{itemize}

\noindent
The operations above are key MPI operations. Complex operations, such as \verb"MPI_Bcast" and \verb"MPI_Gather", can be implemented by composing these key operations. The formal semantics of the core language is defined based on communicating state machines (CSM) \cite{brand1983communicating}. We define each process as a CSM with an unbounded receiving FIFO queue.
For the sake of space limit, the formal semantics can be referred to \cite{mpisv-arxiv}.

An MPI program runs in many processes spanned across multiple machines. These processes communicate by message passing to accomplish a parallel task. Besides parallel execution, the non-determinism in MPI programs mainly comes from two sources: (1) inputs, which may influence the communication through control flow, and (2) wildcard receives, which lead to highly non-deterministic executions.
\begin{figure}
\small
\begin{center}
{\begin{tabular}{l|l|l|l}
\hline
$P_0$&$P_1$&$P_2$&$P_3$\\
\hline
\verb"Send(1)"&\textbf{if} ($x$ != `a') & \verb"Send(1)" & \verb"Send(1)"\\
& \ \ \verb"Recv(0)"& &\\
& \textbf{else}& &\\
& \ \ \verb"IRecv(*,req)";& &\\
& \verb"Recv(3)"& &\\
\hline
\end{tabular}}
\end{center}
\caption{An illustrative example of MPI programs.}
\label{fig:example}
\end{figure}

Consider the MPI program in Figure~\ref{fig:example}. Processes $P_0$, $P_2$ and $P_3$ only send a message to $P_1$ and then terminate. For process $P_1$, if input $x$ is \emph{not} equal to `a', $P_1$ receives a message from $P_0$ in a blocking manner; otherwise, $P_1$ uses a non-blocking wildcard receive to receive a message. Then, $P_1$ receives a message from $P_3$. When $x$ is `a' and \verb"IRecv(*,req)" receives the message from $P_3$, a \emph{deadlock} occurs, \ie, $P_1$ blocks at \verb"Recv(3)", and all the other processes terminate. Hence, to detect the deadlock, we need to handle the non-determinism caused by the input $x$ and the wildcard receive \verb"IRecv(*,req)".

To handle non-determinism due to the input, 
a standard remedy is symbolic execution~\citep{king1976symbolic}. However, there are two challenges. The first one is to \emph{systematically explore the paths of an MPI program with non-blocking and wildcard operations}, which significantly increase the complexity of MPI programs. A non-blocking operation does not block but returns immediately, causing out-of-order completion. The difficulty in handling wildcard operations is to get all the possibly matched messages. The second one is to \emph{improve the scalability of the symbolic execution.} Symbolic execution struggles with path explosion. 
MPI processes run concurrently, resulting in an exponential number of program paths \wrt the number of processes. Furthermore, the path space increases exponentially with the number of wildcard operations.



\subsection{Our Approach}

\mpise\ leverages dynamic verification~\citep{DBLP:conf/cav/VakkalankaGK08} and model checking~\citep{clarke1999model} to tackle the challenges. Figure~\ref{fig:framework} shows \mpise's basic framework.
The inputs of \mpise\ are an MPI program and an expected property, \emph{e.g.}, \emph{deadlock freedom} expressed in LTL. \mpise\ uses the built-in symbolic executor to explore the path space automatically and checks the property along with path exploration. For a path that violates the property, called a \emph{violation path}, \mpise\ generates a test case for replaying, which includes the program inputs, the interleaving sequence of MPI operations and the matchings of wildcard receives. In contrast, for a \emph{violation-free} path $p$, \mpise\ builds a communicating sequential process (CSP) model $\Gamma$, which represents 
the paths which can be obtained based on $p$ by changing the interleavings and matchings of the communication operations in $p$. Then, \mpise\ utilizes a CSP model checker to verify $\Gamma$ \wrt the property. If the model checker reports a counterexample, a violation is found; otherwise, if $\Gamma$ satisfies the property, \mpise\ prunes 
all behaviors captured by the model so that they are avoided by symbolic execution.
\begin{figure}[!t]
\includegraphics[width=3.3in]{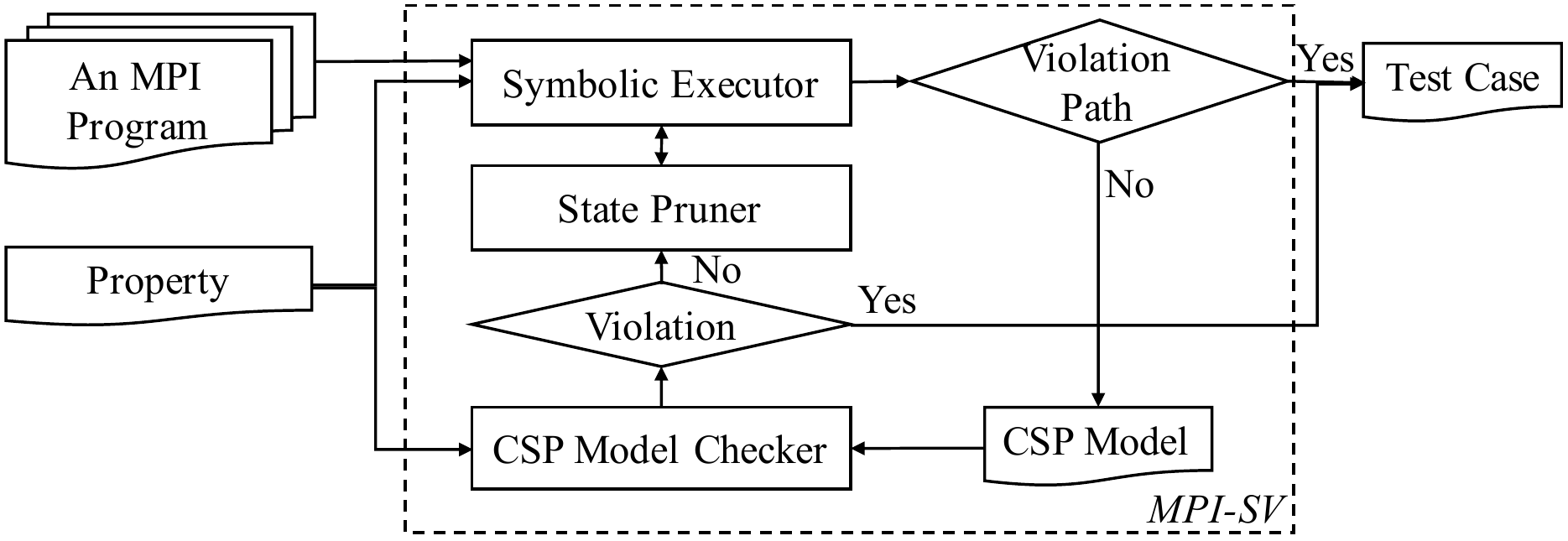}
\caption{The framework of \mpise.}\label{fig:framework}
\end{figure}


Since MPI processes are memory independent, \mpise\ will select a process to execute in a \emph{round-robin} manner to avoid exploring all interleavings of the processes. A process keeps running until it \emph{blocks} or \emph{terminates}. When encountering an MPI operation, \mpise\ records the operation instead of executing it and doing the message matching. When every process blocks or terminates and at least one blocked process exists, \mpise\ matches the recorded MPI operations of the processes \wrt the MPI standard \cite{MPI}.
The intuition behind this strategy is to collect the message exchanges as thoroughly as possible, which helps find possible matchings for the wildcard receive operations. Consider the MPI program in Figure~\ref{fig:example} and the \emph{deadlock freedom} property. Figure \ref{fig:tree} shows the symbolic execution tree, where the node labels indicate process communications, \eg, $(3, 1)$ means that $P_1$ receives a message from $P_3$.
\mpise\ first symbolically executes $P_0$, which only sends a message to $P_1$.
The \verb"Send(1)" operation returns immediately with the assumption of infinite system buffers. Hence, $P_0$ terminates, and the operation \verb"Send(1)" is recorded. Then, \mpise\ executes $P_1$ and explores both branches of the conditional statement as follows.
\begin{figure}[!b]
\includegraphics[width=1.8in]{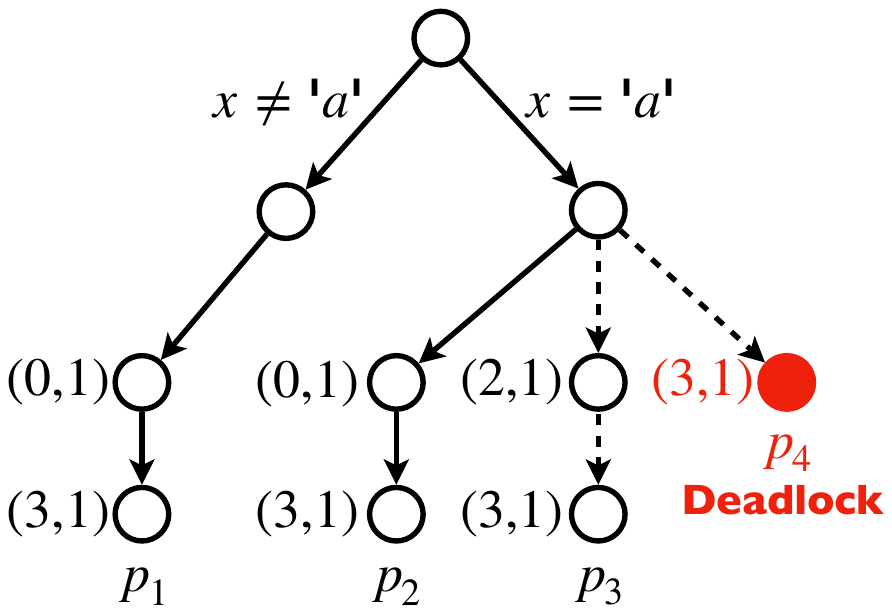}
\caption{The example program's symbolic execution tree.}\label{fig:tree}
\end{figure}

\textbf{(1) True branch ($\mathbf{x} \neq $  `a').}\quad In this case, $P_1$ blocks at \verb"Recv(0)". \mbox{\mpise}\ records the receive operation for $P_1$, and starts executing $P_2$. Like $P_0$, $P_2$ executes operation \verb"Send(1)" and terminates, after which $P_3$ is selected and behaves the same as $P_2$. After $P_3$ terminates, the global execution blocks, \ie, $P_1$ blocks and all the other processes terminate. When this  happens, \mpise\ matches the recorded operations, performs the message exchanges and continues to execute the matched processes. The \verb"Recv(0)" in $P_1$ should be matched with the \verb"Send(1)" in $P_0$. After executing the send and receive operations, \mpise\ selects $P_1$ to execute, because $P_0$ terminates. Then, $P_1$ blocks at \verb"Recv(3)". Same as earlier, the global execution blocks and operation matching needs to be done. \verb"Recv(3)" is matched with the \verb"Send(1)" in $P_3$. After executing the \verb"Recv(3)" and \verb"Send(1)" operations, all the processes terminate successfully. Path $p_1$ in Figure \ref{fig:tree} is explored.

\textbf{(2) False branch ($\mathbf{x}$ {=}`a').}\quad The execution of $P_1$ proceeds until reaching the blocking receive \verb"Recv(3)". Additionally, the two issued receive operations, \ie, \verb"IRecv(*,req)" and \verb"Recv(3)", are recorded. Similar to the true branch, when every process blocks or terminates, we handle operation matching. Here $P_0$, $P_2$ and $P_3$ terminate, and $P_1$ blocks at \verb"Recv(3)".
\verb"IRecv(*,req)" should be matched first because of the \emph{non-overtaken} policy in the MPI standard~\cite{MPI}. There are three \verb"Send" operation candidates from $P_0$, $P_2$ and $P_3$, respectively. \mpise\ forks a state for each candidate. Suppose \mpise\ first explores the state where \verb"IRecv(*,req)" is matched with $P_0$'s \verb"Send(1)".
After matching and executing $P_1$'s \verb"Recv(3)" and $P_3$'s \verb"Send(1)", the path terminates successfully, which generates path $p_2$ in Figure \ref{fig:tree}.

\textbf{\emph{Violation detection}.}\quad \mpise\ continues to explore the remaining two cases. Without \hbox{CSP-based} boosting, the deadlock would be found in the last case (\ie, $p_4$ in Figure \ref{fig:tree}), where \verb"IRecv(*,req)" is matched with $P_3$'s \verb"Send(1)" and $P_1$ blocks because \verb"Recv(3)" has no matched operation. \mpise\ generates a CSP model $\Gamma$ based on the deadlock-free path $p_2$ where $P_1$'s \verb"IRecv(*,req)" is matched with $P_0$'s \verb"Send(1)". Each MPI process is modeled as a CSP process, and all the CSP processes are composed in parallel to form $\Gamma$. Notably, in $\Gamma$, we collect the possible matchings of a wildcard receive through statically matching the arguments of operations in the path. Additionally, the requirements in the MPI standard, \emph{i.e.}, completes-before relations~\citep{DBLP:conf/cav/VakkalankaGK08}, are also modeled. A CSP model checker then verifies deadlock freedom for $\Gamma$. The model checker reports a counterexample where \verb"IRecv(*,req)" is matched with the \verb"Send(1)" in $P_3$. \mpise\ only explores \emph{two} paths for detecting the deadlock and avoids the exploration of $p_3$ and $p_4$ (indicated by dashed lines).

\textbf{\emph{Pruning}.}\quad Because the CSP modeling is precise (\emph{cf.} Section \ref{boosting}), in addition to finding violations earlier, \mpise\ can also perform path pruning when the model satisfies the property. Suppose we change the program in \mbox{Figure \ref{fig:example}} to be the one where the last statement of $P_1$ is a \verb"Recv(*)" operation. Then, the program is \emph{deadlock free}. The true branch ($x \neq$ `a') has 2 paths, because the last wildcard receive in $P_1$ has two matchings (\ie, $P_2$'s send and $P_3$'s send, and $P_0$'s send has been matched by $P_1$'s \verb"Recv(0)"). The false branch ($x =$ `a') has 6 paths because the first wildcard receive has 3 matchings (send operations from $P_0$, $P_2$ and $P_3$) and the last wildcard receive has 2 matchings (because the first wildcard receive has matched one send operation).  Hence, in total, there are 8 paths (\ie, $2 + 3*2 = 8$) if we use pure symbolic execution. In contrast, with model checking, \mpise\ only needs 2 paths to verify that the program is deadlock-free. For each branch, the generated model is verified to be deadlock-free, so \mpise\ prunes the candidate states forked for the matchings of the wildcard receives.


\textbf{\emph{Properties}.}\quad Because our CSP modeling encodes the interleavings of the MPI operations in the MPI processes, the scope of the verifiable properties is enlarged, \ie, \mpise\ can verify safety and liveness properties in LTL. Suppose we change the property to be the one that requires the \verb"Send(1)" operation in $P_0$ should be completed before the \verb"Send(1)" operation in $P_2$. Actually, the send operation in $P_2$ can be completed before the send operation in $P_0$, due to the nature of parallel execution. However, \emph{pure} symbolic execution fails to detect the property violation. In contrast, with the help of CSP modeling, when we verify the model generated from the first path \wrt the property, the model checker gives a counterexample, indicating that 
a violation of the property exists. 

\section{Symbolic Verification Method }\label{sec3}

In this section, we present our symbolic verification framework and then describe \mpise's symbolic execution method.

\subsection{Framework}

Given an MPI program $\mathcal{MP} = \{\textsf{Proc}_i \mid 0 \le i \le n \}$, a state $S_c$ in $\mathcal{MP}$'s symbolic execution is composed by the states of processes, \ie, $(s_0, ..., s_n)$, and each MPI process's state is a 6-tuple $(\mathcal{M}, \mathit{Stat}, \mathit{PC}, \mathcal{F},\mathcal{B},\mathcal{R})$, where $\mathcal{M}$ maps each variable to a concrete value or a symbolic value, $\mathit{Stat}$ is the next program statement to execute, $\mathit{PC}$ is the process's path constraint \cite{king1976symbolic}, $\mathcal{F}$ is the flag of process status belonging to $\{\mathsf{active},\mathsf{blocked},\mathsf{terminated}\}$, $\mathcal{B}$ and $\mathcal{R}$ are infinite buffers for storing the issued MPI operations not yet matched and the matched MPI operations, respectively. We use $s_i \in S_c$ to denote that $s_i$ is a process state in the global state $S_c$. An element $elem$ of $s_i$ can be accessed by $s_i.elem$, \eg, $s_i.\mathcal{F}$ is the $i$th process's status flag. In principle, a statement execution in any process advances the global state, making $\mathcal{MP}$'s state space exponential to the number of processes. We use variable $Seq_i$ defined in $\mathcal{M}$ to record the sequence of the issued MPI operations in $\textsf{Proc}_i$, and $\mathsf{Seq}(S_c)$ to denote the set $\{Seq_i \mid 0 \le i \le n\}$ of global state $S_c$. Global state $S_c$'s path condition (denoted by $S_c.PC$) is the conjunction of the path conditions of $S_c$'s processes, \emph{i.e.}, $\bigwedge\limits_{s_i \in S_c} s_i.PC$.

\begin{spacing}{1}
{\SetAlgoNoLine
\begin{algorithm}[!t]
\setstretch{0.85}
\caption{Symbolic Verification Framework}
\label{alg:se_top}
\LinesNumbered
\DontPrintSemicolon
$\textsf{\mpise}(\mathcal{MP}, \varphi, \textsf{Sym})$\\
\KwData{$\mathcal{MP}$ is $\{\textsf{Proc}_i \mid 0 \le i \le n \}$, $\varphi$ is a property, and $\textsf{Sym}$ is a set of symbolic variables}
\Begin{
$\mathit{worklist} \leftarrow \{S_{init}\}$\\
\While{$\mathit{worklist} \neq \emptyset$}{ \label{while}
$S_c \leftarrow \textsf{Select}(\mathit{worklist})$\label{searcher}\\
$(\mathcal{M}_i, Stat_i, PC_i, \mathcal{F}_i,\mathcal{B}_i,\mathcal{R}_i) \leftarrow \textsf{Scheduler}(S_c)$\\
$\textsf{Execute}(S_c, \textsf{Proc}_i, Stat_i, \textsf{Sym}, \mathit{worklist})$\label{SExe}\\
\If{$\forall s_i \in S_c, s_i.\mathcal{F}= \mathsf{terminated}$} {
$\Gamma \leftarrow \textsf{GenerateCSP}(S_c)$\label{generate_model} \\
$\textsf{ModelCheck}(\Gamma,\varphi)$\label{MC}\\
\If{$\Gamma\models \varphi$}{
$\mathit{worklist} {\leftarrow} \mathit{worklist} {\setminus} \{ S_p {\in} \mathit{worklist} {\mid} S_p.PC {\Rightarrow} S_c.PC\}$\label{prune}
}\ElseIf{$\Gamma \not\models \varphi$} {
$\textsf{reportViolation}$ and \textbf{Exit}\label{report}
}
}
}
}
\end{algorithm}}
\end{spacing}

\mbox{Algorithm \ref{alg:se_top}} shows the details of \mpise.
We use $\mathit{worklist}$ to store the global states to be explored. Initially, $\mathit{worklist}$ only contains $S_{init}$, composed of the initial states of all the processes, and each process's status is $\mathsf{active}$. At Line \ref{searcher}, \textsf{Select} picks a state from $\mathit{worklist}$ as the one to advance. Hence, \textsf{Select} can be customized with different search heuristics, \eg, depth-first search (DFS). Then, \textsf{Scheduler} selects an active process $\textsf{Proc}_i$ to execute. Next, \textsf{Execute} (\cf Algorithm~\ref{alg:se_statement}) symbolically executes the statement $Stat_i$ in $\textsf{Proc}_i$, and may add new states into $\mathit{worklist}$. This procedure continues until $\mathit{worklist}$ is empty (\ie, all the paths have been explored), detecting a violation or time out (omitted for brevity). After executing $Stat_i$, if all the processes in the current global state $S_c$ terminate, \emph{i.e.}, a violation-free path terminates, we use Algorithm \ref{alg:Modeling} to generate a CSP model $\Gamma$ from the current state (Line~\ref{generate_model}). Then, we use a CSP model checker to verify $\Gamma$ \wrt $\varphi$. If $\Gamma$ satisfies $\varphi$ (denoted by $\Gamma\models \varphi$), we prune the global states forked by the wildcard operations along the current path (Line \ref{prune}), \emph{i.e.}, the states in $\mathit{worklist}$ whose path conditions imply $S_c$'s path condition; otherwise, if the model checker gives a counterexample, we report the violation and exit (Line \ref{report}).

Since MPI processes are memory independent, we employ partial order reduction (POR)~\citep{clarke1999model} to reduce the search space. \textsf{Scheduler} selects a process in a \emph{round-robin} fashion from the current global state. In principle, \textsf{Scheduler} starts from the active MPI process with the smallest identifier, \emph{e.g.}, $\textsf{Proc}_0$ at the beginning, and an MPI process keeps running until it is blocked or terminated. Then, the next active process will be selected to execute. Such a strategy significantly reduces the path space of symbolic execution.
Then, with the help of CSP modeling and model checking, \mpise\ can verify more properties, \ie, safety and liveness properties in LTL. The details of such technical improvements will be given in Section~\ref{boosting}.

\subsection{\LAZY\ Symbolic Execution}

Algorithm~\ref{alg:se_statement} shows the symbolic execution of a statement.
Common statements such as conditional statements are
handled in the standard way~\citep{king1976symbolic} (omitted for brevity), and here we focus on MPI operations.
{\SetAlgoNoLine
\begin{algorithm}[!t]
\setstretch{0.85}
\caption{Blocking-driven Symbolic Execution}
\label{alg:se_statement}
\LinesNumbered
\DontPrintSemicolon
$\textsf{Execute}(S_c, \textsf{Proc}_i, Stat_i, \textsf{Sym}, worklist)$\\
\KwData{Global state $S_c$, MPI process $\textsf{Proc}_i$, Statement $Stat_i$, Symbolic variable set $\textsf{Sym}$, $\mathit{worklist}$ of global states }
\Begin{
\Switch{$(Stat_i)$}{
\Case{$\emph{\texttt{Send}}\ or\ \emph{\texttt{ISend}}\ or\ \emph{\texttt{IRecv}}$}{
$Seq_i \leftarrow Seq_i\cdot \langle Stat_i\rangle$\label{record1}\\
$s_i.\mathcal{B} \leftarrow s_i.\mathcal{B}\cdot \langle Stat_i\rangle$\label{buf1}
}
\Case{$\emph{\texttt{Barrier}}\ or\ \emph{\texttt{Wait}}\ or\ \emph{\texttt{Ssend}}\ or\ \emph{\texttt{Recv}}$}{
$Seq_i \leftarrow Seq_i\cdot \langle Stat_i\rangle$\label{record2} \\
$s_i.\mathcal{B} \leftarrow s_i.\mathcal{B}\cdot \langle Stat_i\rangle$\label{buf2} \\
$s_i.\mathcal{F} \leftarrow \mathsf{blocked}$\label{blocks}\\
\If{\emph{\textsf{GlobalBlocking}}}{ \label{LMC} \tcp*{{\footnotesize$\forall s_i \in S_c, (s_i.\mathcal{F}=\mathsf{blocked} \vee s_i.\mathcal{F} =\mathsf{terminated})$}}
$\textsf{Matching}(S_c, worklist)$\label{LM}
}
}
\textbf{default:} $\textsf{Execute}(S_c, \textsf{Proc}_i, Stat_i, \textsf{Sym}, worklist)\ as\ normal$
}
}
\end{algorithm}
}
The main idea is to \emph{delay} the executions of MPI operations \emph{as much as possible}, \emph{i.e.}, trying to get all the message matchings. Instead of execution, \mbox{Algorithm~\ref{alg:se_statement}} records each MPI operation for each MPI process (Lines~\ref{record1}\&\ref{record2}). We also need to update buffer $\mathcal{B}$ after issuing an MPI operation (Lines~\ref{buf1}\&\ref{buf2}).
Then, if $Stat_i$ is a non-blocking operation, the execution returns immediately; otherwise, we block $\textsf{Proc}_i$ (Line \ref{blocks}, excepting the \verb"Wait" of an \verb"ISend" operation). When reaching $\textsf{GlobalBlocking}$ (Lines~\ref{LMC}\&\ref{LM}), \emph{i.e.}, every process is terminated or blocked, we use $\textsf{Matching}$ (\cf Algorithm \ref{alg:matching}) to match the recorded but not yet matched MPI operations and execute the matched operations. Since the opportunity of matching messages is $\textsf{GlobalBlocking}$, we call it \lazy\ symbolic execution.

{\SetAlgoNoLine
\begin{algorithm}[!b]
\setstretch{0.85}
\caption{\LAZY\ Matching}
\label{alg:matching}
\LinesNumbered
\DontPrintSemicolon
$\textsf{Matching}(S_c, \mathit{worklist})$\\
\KwData{Global state $S_c$, $\mathit{worklist}$ of global states}
\Begin{
$MS_W \leftarrow \emptyset$\tcp*{Matching set of wildcard operations}
$\mathit{pair}_n \leftarrow \textsf{match}_{N}(S_c)$\label{matchN}\tcp*{Match non-wildcard operations}
\If{$\mathit{pair}_n \neq \mathit{empty\ pair}$}{
$\textsf{Fire}(S_c, pair_n )$\label{fire}
}\Else{
$MS_W \leftarrow \textsf{match}_{W}(S_c)$\label{matchW}\tcp*{Match wildcard operations}
\For{$\mathit{pair}_w \in MS_{W}$}{ \label{updateS}
$S_c'\leftarrow \textsf{fork}(S_c, \mathit{pair}_w)$\label{fork}\\
$\mathit{worklist} \leftarrow \mathit{worklist} \cup \{S_c'\}$\label{forkE}
}
\If{$MS_W \neq \emptyset$}{
$\mathit{worklist} \leftarrow \mathit{worklist} \setminus \{S_c\}$
}\label{updateE}
}
\If{$\mathit{pair}_n = \mathit{empty\ pair} \wedge MS_W = \emptyset$}{\label{deadlock}
$\textsf{reportDeadlock}$ and \textbf{Exit}\label{exit}
}
}
\end{algorithm}
}

$\textsf{Matching}$ matches the recorded MPI operations in different processes.
To obtain all the possible matchings, we delay the matching of a wildcard operation \emph{as much as possible}. We use $\textsf{match}_{N}$ to match the non-wildcard operations first (Line \ref{matchN}) \wrt the rules in the MPI standard \cite{MPI}, especially the \emph{non-overtaken} ones: (1) if two sends of a process send messages to the same destination, and both can match the same receive, the receive should match the first one; and (2) if a process has two receives, and both can match a send, the first receive should match the send. The matched send and receive operations will be executed, and the statuses of the involved processes will be updated to \textsf{active}, denoted by $\textsf{Fire}(S_c, \mathit{pair}_n)$ (Line \ref{fire}). If there is no matching for non-wildcard operations, we use $\textsf{match}_W$ to match the wildcard operations (Line~\ref{matchW}). For each possible matching of a wildcard receive, we fork a new state (denoted by $\textsf{fork}(S_c, \mathit{pair}_w)$ at Line \ref{fork}) to analyze each matching case. If no operations can be matched, but there exist blocked processes, a deadlock happens (Line \ref{deadlock}). Besides, for the LTL properties other than deadlock freedom (such as temporal properties), we also check them during symbolic execution (omitted for brevity).

Take the program in \mbox{Figure~\ref{fig:example_lazy}} for example. When all the processes block at \verb"Barrier", \mpise\ matches the recorded operation in the buffers of the processes, \ie,
{\small$s_0.\mathcal{B}{=}\langle\verb"ISend(1,req"_1){,} \verb"Barrier"\rangle$}, {\small$s_1.\mathcal{B}{=}\langle\verb"IRecv(*,req"_2\verb")", \verb"Barrier"\rangle$}, and {\small$s_2.\mathcal{B}{=}\langle\verb"Barrier"\rangle$}. According to the MPI standard, each operation in the buffers is ready to be matched. Hence, $\textsf{Matching}$ first matches the non-wildcard operations, \ie, the {\small\verb"Barrier"} operations, then the status of each process becomes $\mathsf{active}$. After that, \mpise\ continues to execute the active processes and record issued MPI operations. The next $\textsf{GlobalBlocking}$ point is: $P_0$ and $P_2$ terminate, and $P_1$ blocks at \verb"Wait(req"$_2$\verb")". The buffers are {\small$\langle\verb"ISend(1,req"_1\verb")"{,} \verb"Wait(req"_1\verb")"\rangle$}, {\small$\langle\verb"IRecv(*,req"_2\verb")"{,} \verb"Wait(req"_2\verb")"\rangle$}, and {\small$\langle\verb"ISend(1,req"_3\verb")", \verb"Wait(req"_3\verb")"\rangle$}, respectively. All the issued \verb"Wait" operations are not ready to match, because the corresponding non-blocking operations are not matched. So $\textsf{Matching}$ needs to match the wildcard operation, \ie, $\verb"IRecv(*,req"_2\verb")"$, which can be matched with $\verb"ISend(1,req"_1\verb")"$ or $\verb"ISend(1,req"_3\verb")"$. Then, a new state is forked for each case and added to the $\mathit{worklist}$.

\begin{figure}
\begin{center}
{\begin{tabular}{l|l|l}\hline
$P_0$&$P_1$&$P_2$\\ \hline
\verb"ISend(1,req"$_1$\verb")"; & \verb"IRecv(*,req"$_2$\verb")"; & \verb"Barrier";\\
\verb"Barrier"; & \verb"Barrier"; & \verb"ISend(1,req"$_3$\verb")";\\
\verb"Wait(req"$_1$\verb")" & \verb"Wait(req"$_2$\verb")" & \verb"Wait(req"$_3$\verb")"\\
\hline
\end{tabular}}
\end{center}
\caption{An example of operation matching.}
\label{fig:example_lazy}
\end{figure}

\textbf{Correctness}. Blocking-driven symbolic execution is an instance of model checking with POR.
We have proved the symbolic execution method is correct for \emph{reachability properties} \citep{DBLP:books/daglib/0077033}. Due to the space limit, the proof can be referred to \cite{mpisv-arxiv}.%


\section{CSP Based Path Modeling}\label{boosting}

In this section, we first introduce the CSP~\cite{roscoe1998theory} language. Then, we present the modeling algorithm of an MPI program terminated path using a subset of CSP.
Finally, we prove the soundness and completeness of our modeling.
\subsection{CSP Subset}

Let $\Sigma$ be a {\em finite} set of {\em events}, $\mathbb{C}$ a set of channels, and $\mathbf{X}$ a set of variables. Figure~\ref{CSPSubset} shows the syntax of the CSP subset, where $P$ denotes a CSP process, $a {\in} \Sigma$, $c {\in} \mathbb{C}$, $X {\subseteq} \Sigma$ and $x {\in} \mathbf{X}$.
\noindent
\begin{figure}[!hbt]
\begin{center}
{
\begin{tabular}{l}
$P := a \mid P\ {\fatsemi}\ P \mid P \square P \mid P {\underset{X}{\parallel}} P \mid c?x {\rightarrow} P \mid c!x {\rightarrow} P \mid \textbf{skip} $
\end{tabular}
}
\end{center}
\caption{The syntax of a CSP subset.}\label{CSPSubset}
\end{figure}

The single event process $a$ performs the event $a$ and terminates. There are three operators: sequential composition ($\fatsemi$), external choice ($\square$) and parallel composition with synchronization ($\underset{X}{\parallel}$). $P\square Q$ performs as $P$ or $Q$, and the choice is made by the environment. Let $PS$ be a finite set of processes, $\square PS$ denotes the external choice of all the processes in $PS$. $P\underset{X}{\parallel} Q$ performs $P$ and $Q$ in an interleaving manner, but $P$ and $Q$ synchronize on the events in $X$. The process $c?x \rightarrow P$ performs as $P$ after reading a value from channel $c$ and writing the value to variable $x$. The process $c!x \rightarrow P$ writes the value of $x$ to channel $c$ and then behaves as $P$. Process $\textbf{skip}$ terminates immediately.

\subsection{CSP Modeling}\label{cspmodeling}

For each violation-free program path, Algorithm \ref{alg:Modeling} builds a precise CSP model of the possible communication behaviors by changing the matchings and interleavings of the communication operations along the path. 
{\SetAlgoNoLine
{
\begin{algorithm}[!b]
\setstretch{0.85}
\caption{CSP Modeling for a Terminated State}
\label{alg:Modeling}
\LinesNumbered
\DontPrintSemicolon
$\textsf{GenerateCSP}(S)$\\
\KwData{A terminated global state $S$, and $\mathsf{Seq}(S) {=} \{Seq_i \mid 0 \le i \le n\}$}
\Begin{
$PS \leftarrow \emptyset$ \\
\For{$i \leftarrow 0 \ \dots\ n\ $}{
$P_i \leftarrow \textbf{skip}$\\
$Req \leftarrow \{r \mid \texttt{IRecv(*,r)} {\in} Seq_i {\vee} \texttt{IRecv(i,r)} {\in} Seq_i\}$\\
\For{$j \leftarrow\!length(Seq_i)-1\,\dots\,0\ $}{
   \Switch{$op_j$}{
   \Case{\emph{\texttt{Ssend(i)}}}{
   $c_1 \leftarrow \textsf{Chan}(op_j)$   \tcp*{$c_1$'s size is 0}
   $P_i \leftarrow c_1!x \rightarrow P_i$\label{ssend} \\
   }
   \Case{\emph{\texttt{Send(i)}}\ or\ \emph{\texttt{ISend(i,r)}}}{ \label{sendS}
   $c_2 \leftarrow \textsf{Chan}(op_j)$ \tcp*{$c_2$'s size is 1}
   $P_i \leftarrow c_{2}!x \rightarrow P_i$\label{send} \\
   }
   \Case{\emph{\texttt{Barrier}}}{ \label{BarrierS}
   $P_i \leftarrow $\ \texttt{B}$\ \fatsemi\ P_i$ \label{barrier} \label{BarrierE}\\
   }
   \Case{\emph{\texttt{Recv(i)}}\ or\ \emph{\texttt{Recv(*)}}}{
   $C \leftarrow \textsf{StaticMatchedChannel}(op_j, S)$\label{static1}\\
   $Q \leftarrow \textsf{Refine} (\square \{c?x\rightarrow \textbf{skip} \mid c \in C\}, S)$\label{choice1}\\
   $P_i \leftarrow Q \fatsemi P_i$\label{seq} \\
   }
   \Case{\emph{\texttt{IRecv(*,r)}}\ or \ \emph{\texttt{IRecv(i,r)}}}{
   $C \leftarrow \textsf{StaticMatchedChannel}(op_j, S)$\label{static2} \\
   $Q \leftarrow \textsf{Refine} (\square \{c?x\rightarrow \textbf{skip} \mid c \in C\}, S)$\label{choice2}\\
   $e_w {\leftarrow} \textsf{WaitEvent}(op_j)$ \label{choiceW}\hspace{-1mm}\tcp*{$op_j$'s wait event}
   $P_i \leftarrow (Q \fatsemi e_w) \underset{\{e_w\}}{\parallel}\ P_i$ \label{para} \\
   }
   \Case{\emph{\texttt{Wait(r)} and $r \in Req$}
   }{ \label{wait}
   $e_w \leftarrow \textsf{GenerateEvent}(op_j)$\label{waitevent}\\
   $P_i \leftarrow e_w \fatsemi P_i$ \label{waitE}
   }
   }
}
   $PS \leftarrow PS \cup \{P_i\}$
   }
   $P \leftarrow \underset{\{\texttt{B}\}}{\parallel}PS$\label{compose}\label{parallel}\\
    \Return $P$
   }
\end{algorithm}}
}
The basic idea is to model the communication operations in each process as a CSP process, then compose all the CSP processes in parallel to form the model. To model $\textsf{Proc}_i$, we scan its operation sequence $Seq_i$ in reverse. For each operation, we generate its CSP model and compose the model with that of the remaining operations in $Seq_i$ \wrt the semantics of the operation and the MPI standard~\cite{MPI}. The modeling algorithm is efficient, and has a polynomial time complexity \wrt the total length of the recorded MPI operation sequences.

We use channel operations in CSP to model send and receive operations. Each send operation $op$ has its own channel, denoted by $\textsf{Chan}(op)$. We use a \emph{zero-sized} channel to model \verb"Ssend" operation (Line \ref{ssend}), because \verb"Ssend" blocks until the message is received. In contrast, considering a \verb"Send" or \verb"ISend" operation is completed immediately, we use \emph{one-sized} channels for them (Line~\ref{send}), so the channel writing returns immediately. The modeling of \verb"Barrier"  (Line~\ref{barrier}) is to generate a synchronization event that requires all the parallel CSP processes to synchronize it (Lines~\ref{BarrierE}\&\ref{parallel}). The modeling of receive operations consists of three steps. The first step calculates the possibly matched channels written by the send operations (Lines~\ref{static1}\&\ref{static2}). The second uses the external choice of reading actions of the matched channels (Lines~\ref{choice1}\&\ref{choice2}), so as to model different cases of the receive operation. Finally, the refined external choice process is composed with the remaining model. If the operation is blocking, the composition is sequential (Line \ref{seq}); otherwise, it is a parallel composition (Line \ref{para}).

$\textsf{StaticMatchedChannel}(op_j, S)$ (Lines~\ref{static1}\&\ref{static2}) returns the set of the channels written by the possibly matched send operations of the receive operation $op_j$.
We scan $\mathsf{Seq}(S)$ to obtain the possibly matched send operations of $op_j$. Given a receive operation $recv$ in process $\textsf{Proc}_i$,  $\textsf{SMO}(recv, S)$ calculated as follows denotes the set of the matched send operations of $recv$.
\begin{itemize}[leftmargin=1em]
\item If $recv$ is $\verb"Recv"(j)$ or $\verb"IRecv"(j, \verb"r")$, $\textsf{SMO}(recv,S)$ contains $\textsf{Proc}_j$'s send operations with $\textsf{Proc}_i$ as the destination process.
\item If $recv$ is $\verb"Recv"(*)$ or $\verb"IRecv"(*,\verb"r")$, $\textsf{SMO}(recv,S)$ contains \emph{any process}'s send operations with $\textsf{Proc}_i$ as the destination process.
\end{itemize}

$\textsf{SMO}(op, S)$ over-approximates $op$'s precisely matched operations, and can be optimized by removing the send operations that are definitely executed after $op$'s completion, and the ones whose messages are definitely received before $op$'s issue.
For example, Let $\mathsf{Proc}_0$ be {\small\verb"Send(1)"$;$\verb"Barrier"$;$\verb"Send(1)"}, and $\mathsf{Proc}_1$ be {\small\verb"Recv(*)"$;$\verb"Barrier"}. $\textsf{SMO}$ will add the two send operations in $\mathsf{Proc}_0$ to the matching set of the \verb"Recv(*)" in $\mathsf{Proc}_1$. Since \verb"Recv(*)" must complete before $\verb"Barrier"$, we can remove the second send operation in $\mathsf{Proc}_0$.
Such optimization reduces the complexity of the CSP model. For brevity, we use $\textsf{SMO}(op, S)$ to denote the optimized matching set. Then, $\textsf{StaticMatchedChannel}(op_j, S)$ is
$
\{ \textsf{Chan}(op) \mid op \in \textsf{SMO}(op_j, S) \}
$.

To satisfy the MPI requirements, $\textsf{Refine}(P, S)$ (Lines~\ref{choice1}\&\ref{choice2}) refines the models of receive operations by imposing the completes-before requirements~\cite{DBLP:conf/cav/VakkalankaGK08} as follows:
\begin{itemize}[leftmargin=1em]
\setlength{\itemsep}{0pt}\setlength{\parsep}{0pt}\setlength{\parskip}{0pt}
\item If a receive operation has multiple matched send operations from the same process, it should match the earlier issued one. This is ensured by checking the emptiness of the dependent channels.
\item The receive operations in the same process should be matched \wrt their issue order if they receive messages from the same process, except the \emph{conditional completes-before} pattern \cite{DBLP:conf/cav/VakkalankaGK08}. We use one-sized channel actions to model these requirements.
\end{itemize}

We model a \verb"Wait" operation if it corresponds to an \verb"IRecv" operation (Line \ref{wait}), because \verb"ISend" operations complete immediately under the assumption of infinite system buffer.
\verb"Wait" operations are modeled by the synchronization in parallel processes. $\textsf{GenerateEvent}$ generates a new synchronization event $e_w$ for each \verb"Wait" operation (Line \ref{waitevent}). Then, $e_w$ is produced after the corresponding non-blocking operation is completed (Line \ref{para}). The synchronization on $e_w$ ensures that a \verb"Wait" operation blocks until the corresponding non-blocking operation is completed.

We use the example in Figure~\ref{fig:example_lazy} for a demonstration. After exploring a \hbox{violation-free} path, the recorded operation sequences are {\small$Seq_0{=}\langle\verb"ISend(1,req"_1\verb")", \verb"Barrier",\verb"Wait"\verb"(req"_1\verb")"\rangle$}, {\small$Seq_1{=}\langle\verb"IRecv(*,req"_2\verb")",$}\\{\small$ \verb"Barrier"{,} \verb"Wait(req"_2\verb")"\rangle$}, {\small$Seq_2{=}\langle\verb"Barrier"{,} \verb"ISend(1,req"_3\verb")"{,} \verb"Wait(req"_3\verb")"\rangle$}. We first scan $Seq_0$ in reverse. Note that we don't model \verb"Wait(req"$_1$\verb")", because it corresponds to \verb"ISend".
We create a synchronization event \verb"B" for modeling \verb"Barrier" (Lines~\ref{BarrierS}\&\ref{BarrierE}). For the \verb"ISend(1,req"$_1$\verb")", we model it by writing an element $a$ to a one-sized channel $chan_{1}$, and use prefix operation to compose its model with \verb"B" (Lines~\ref{sendS}-\ref{send}).
In this way, we generate CSP process $chan_1!a{\rightarrow} \verb"B"\fatsemi\textbf{skip}$ (denoted by $\mathit{CP}_0$) for $\mathsf{Proc}_0$. Similarly, we model $\mathsf{Proc}_2$ by $\verb"B"\fatsemi chan_2!b{\rightarrow} \textbf{skip}$ (denoted by $\mathit{CP}_2$), where $chan_2$ is also a one-sized channel and $b$ is a channel element. For $\mathsf{Proc}_1$, we generate a single event process $e_w$ to model \verb"Wait(req"$_2$\verb")", because it corresponds to \verb"IRecv" (Lines~\ref{wait}-\ref{waitE}). For \verb"IRecv(*,req"$_2$\verb")", we first compute the matched channels using $\textsf{SMO}$ (Line~\ref{static2}), and $\textsf{StaticMatchedChannel}(op_j, S)$ contains both $chan_1$ and $chan_2$. Then, we generate the following CSP process
\[
((chan_1?a{\rightarrow}\textbf{skip} {\square} chan_2?b{\rightarrow}\textbf{skip})\fatsemi e_w) {\underset{\{e_w\}}{\parallel}} (\verb"B"\fatsemi e_w\fatsemi\textbf{skip})
\]
(denoted by $\mathit{CP}_1$) for $\mathsf{Proc}_1$. Finally, we compose the CSP processes using the parallel operator to form the CSP model (Line~\ref{parallel}), \emph{i.e.}, $\mathit{CP}_0\underset{\{\texttt{B}\}}{\parallel}\ \mathit{CP}_1\underset{\{\texttt{B}\}}{\parallel}\ \mathit{CP}_2$.

CSP modeling supports the case where communications depend on message contents.
\mpise tracks the influence of a message during symbolic execution. When detecting that the message content influences the communications, \mpise symbolizes the content on-the-fly. We specially handle the widely used \emph{master-slave} pattern for dynamic load balancing \cite{Gropp:2014:UMP:2717061}. The basic idea is to use a recursive CSP process to model each slave process and a conditional statement for master process to model the communication behaviors of different matchings. We verified five dynamic load balancing MPI programs in our experiments (\cf Section \ref{exp-result}). The details for supporting master-slave pattern is%
\ifdefined\supplementary
\ in Appendix \ref{master-slave-support}.
\else
\ in the supplementary document.
\fi

\subsection{Soundness and Completeness}\label{sec:prove}

In the following, we show that the CSP modeling is \emph{sound} and \emph{complete}.
Suppose \textsf{GenerateCSP($S$)} generates the CSP process $\textsf{CSP}_s$. Here, \emph{soundness} means that $\textsf{CSP}_s$ models all the possible behaviors by changing the matchings or interleavings of the communication operations along the path to $S$, and \emph{completeness} means that each trace in $\textsf{CSP}_s$ represents a real behavior that can be derived from $S$ by changing the matchings or interleavings of the communications.

Since we compute $\textsf{SMO}(op, S)$ by statically matching the arguments of the recorded operations, $\textsf{SMO}(op, S)$ may contain some false matchings. Calculating the precisely matched operations of $op$ is NP-complete~\citep{forejt2014precise}, and we suppose such an ideal method exists. We use $\textsf{CSP}_{static}$ and $\textsf{CSP}_{ideal}$ to denote the generated models using $\textsf{SMO}(op, S)$ and the ideal method, respectively.
The following theorems ensure the equivalence of the two models under the stable-failure semantics \cite{roscoe1998theory} of CSP and $\textsf{CSP}_{static}$'s consistency to the MPI semantics, which imply the soundness and completeness of our CSP modeling method. 
Let $\mathcal{T}(P)$ denote the trace set~\citep{roscoe1998theory} of CSP process $P$, and $\mathcal{F}(P)$ denote the failure set of CSP process $P$. Each element in $\mathcal{F}(P)$ is $(s, X)$, where $s \in \mathcal{T}(P)$ is a trace, and $X$ is the set of events $P$ refuses to perform after $s$.

\begin{thm}\label{thm:failure}
$\mathcal{F}(\emph{\textsf{CSP}}_{static}) =  \mathcal{F}(\emph{\textsf{CSP}}_{ideal})$.
\end{thm}

\begin{proof} 

We only give the skeleton of the proof. We first prove 
\[\mathcal{T}(\textsf{CSP}_{static}) =  \mathcal{T}(\textsf{CSP}_{ideal})\]
based on which we can prove $\mathcal{F}(\emph{\textsf{CSP}}_{static}) =  \mathcal{F}(\emph{\textsf{CSP}}_{ideal})$. The main idea of proving these two equivalence relations is to use \emph{contradiction} for proving the subset relations. We only give the proof of $\mathcal{T}(\textsf{CSP}_{static}) \subseteq \mathcal{T}(\textsf{CSP}_{ideal})$; the other subset relations can be proved in a similar way. 

Suppose there is a trace $t {=} \langle e_1, ...,  e_n\rangle$ such that $ t \in \mathcal{T}(\textsf{CSP}_{static})$ but $t {\notin} \mathcal{T}(\textsf{CSP}_{ideal})$. The only difference between $\textsf{CSP}_{static}$ and $\textsf{CSP}_{ideal}$ is that $\textsf{CSP}_{static}$ introduces more channel read operations during the modeling of receive operations. Hence, there must exist a read operation of an extra channel in $t$. Suppose the first extra read is $e_k {=} c_e?x$, where $1 \le k \le n$. Therefore, $c_e$ \emph{cannot} be read in $\textsf{CSP}_{ideal}$ when the matching of the corresponding receive operation starts, but $c_e$ is not empty at $e_k$ in $\textsf{CSP}_{static}$. Despite of the size of $c_e$, there must exist a write operation $c_e!y$ in $\langle e_1, ...,  e_{k-1}\rangle$. Because $\langle e_1, ...,  e_{k-1}\rangle$ is also a valid trace in $\textsf{CSP}_{ideal}$, it means $c_e$ is not empty in $\textsf{CSP}_{ideal}$ at $e_k$, which contradicts with the assumption that $c_e$ cannot be read in $\textsf{CSP}_{ideal}$. Hence, $\mathcal{T}(\textsf{CSP}_{static}) \subseteq  \mathcal{T}(\textsf{CSP}_{ideal})$ holds.
\end{proof}

\begin{thm}
$\emph{\textsf{CSP}}_{static}$ is consistent with the MPI semantics.
\end{thm}

The proof's main idea is to prove that $\textsf{CSP}_{ideal}$ is equal to the model defined by the formal MPI semantics \cite{mpisv-arxiv} \wrt the failure divergence semantics. Then, based on Theorem \ref{thm:failure}, we can prove that $\emph{\textsf{CSP}}_{static}$ is consistent with the MPI semantics. Please refer to \cite{mpisv-arxiv} for the detailed proofs for these two theorems.



\section{Experimental Evaluation}\label{experiment}

In this section, we first introduce the implementation of \mpise, then describes the research questions and the experimental setup. Finally, we give experimental results.

\subsection{Implementation}\label{sec:implementation}
We have implemented \mpise
based on Cloud9~\citep{bucur2011parallel}, which is built upon KLEE~\citep{cadar2008klee}, and enhances KLEE with better support for POSIX environment and parallel symbolic execution. 
We leverage Cloud9's support for multi-threaded programs. We use a multi-threaded library for MPI, called AzequiaMPI~\citep{DBLP:conf/pvm/Rico-GallegoM11}, as the MPI environment model for symbolic execution. \mpise\ contains three main modules: program preprocessing, symbolic execution, and model checking. The program preprocessing module generates the input for symbolic execution. We use Clang to compile an MPI program to LLVM bytecode, which is then linked with the pre-compiled MPI library AzequiaMPI. The symbolic execution module is in charge of path exploration and property checking. 
The third module utilizes the state-of-the-art CSP model checker PAT~\citep{sun2009pat} to verify CSP models, and uses the output of PAT to boost the symbolic executor.

\subsection{Research Questions}
We conducted experiments to answer the following questions:
\begin{itemize}[leftmargin=1em]
\setlength{\itemsep}{0pt}\setlength{\parsep}{0pt}\setlength{\parskip}{0pt}
\item Effectiveness: Can \mpise\ verify real-world MPI programs effectively? How effective is \mpise\ when compared to the existing state-of-the-art tools?
\item Efficiency: How efficient is \mpise\ when verifying real-world MPI programs? How efficient is \mpise when compared to the pure symbolic execution?
\item Verifiable properties : Can \mpise\ verify properties other than deadlock freedom?
\end{itemize}

\subsection{Setup}
Table 1 lists the programs analyzed in our experiments. All the programs are real-world open source MPI programs. \texttt{DTG} is a testing program from~\cite{vakkalanka2010efficient}. \texttt{Matmat}, \texttt{Integrate} and \texttt{Diffusion2d} come from the FEVS benchmark suite~\citep{siegel2011fevs}. \texttt{Matmat} is used for matrix multiplication,
\texttt{Integrate} calculates the integrals of trigonometric functions, and \texttt{Diffusion2d} is a parallel solver for two-dimensional diffusion equation. \texttt{Gauss\_elim} is an MPI implementation for gaussian elimination used in~\cite{xue2009mpiwiz}. \texttt{Heat} is a parallel solver for heat equation used in~\cite{muller2011dealing}. \texttt{Mandelbrot}, \texttt{Sorting} and \texttt{Image\_manip} come from github. \texttt{Mandelbrot} parallel draws the mandelbrot set for a bitmap, \texttt{Sorting} uses bubble sort to sort a multi-dimensional array, and \texttt{Image\_manip} is an MPI program for image manipulations, \emph{e.g.}, shifting, rotating and scaling. The remaining three programs are large parallel applications. \texttt{Depsolver} is a parallel multi-material 3D electrostatic solver, \texttt{Kfray} is a ray tracing program creating realistic images, and \texttt{ClustalW} is a tool for aligning gene sequences.

{
\begin{table}[!t]
\caption{The programs in the experiments.}
\label{TABLE:Bechmarks}
\begin{center}
{\begin{tabular}{l|r|l}
\hline
  \textbf{Program}         & \textbf{LOC}            &   \textbf{Brief Description} \\
  \hline
   \texttt{DTG} & $90$ & Dependence transition group \\
  \hline
   \texttt{Matmat} & $105$ & Matrix multiplication \\
  \hline
   \texttt{Integrate} & $181$ & Integral computing \\
  \hline
   \texttt{Diffusion2d} & $197$ & Simulation of diffusion equation \\
  \hline
   \texttt{Gauss\_elim} & $341$ & Gaussian elimination \\
  \hline
   \texttt{Heat} & $613$ & Heat equation solver \\
  \hline
   \texttt{Mandelbrot} & $268$ & Mandelbrot set drawing \\
  \hline
   \texttt{Sorting} & $218$ & Array sorting \\
  \hline
   \texttt{Image\_manip} & $360$ & Image manipulation \\
  \hline
   \texttt{DepSolver} & $8988$ & Multimaterial electrostatic solver \\
  \hline
   \texttt{Kfray} & $12728$ & KF-Ray parallel raytracer \\
  \hline
   \texttt{ClustalW} & $23265$ & Multiple sequence alignment \\
  \hline
   \textbf{Total} & \textbf{47354} & \textbf{12 open source programs} \\
  \hline
\end{tabular}}
\end{center}
\end{table}
}

To evaluate \mpise\ further, we mutate~\citep{just2014mutants} the programs by rewriting a randomly selected receive using two rules: (1) replace {\small$\texttt{Recv(i)}$} with {\small$\textbf{if}\ (x{>}a) \{\texttt{Recv(i)}\}\ \textbf{else}\ \{\texttt{Recv(*)}\}$}; (2) replace {\small$\texttt{Recv(*)}$} with {\small$\textbf{if}\ (x{>}a) \{\texttt{Recv(*)}\}\ \textbf{else}\ \{\texttt{Recv(j)}\}$}. Here $x$ is an input variable, $a$ is a random value, and $j$ is generated randomly from the scope of the process identifier. The mutations for $\texttt{IRecv(i,r)}$ and $\texttt{IRecv(*,r)}$ are similar. Rule 1 is to improve program performance and simplify programming, while rule 2 is to make the communication more deterministic. Since communications tend to depend on inputs in complex applications, such as the last three programs in Table~\ref{TABLE:Bechmarks}, we also introduce input related conditions. For each program, we generate five mutants if possible, or generate as many as the number of receives. We don't mutate the programs using \emph{master-slave} pattern \cite{Gropp:2014:UMP:2717061}, \ie, \texttt{Matmat} and \texttt{Sorting}, and only mutate the static scheduling versions of programs \texttt{Integrate}, \texttt{Mandelbrot}, and \texttt{Kfray}.

\newcommand{\baseline}{\textsf{Symbolic execution}}
\newcommand{\ours}{\textsf{Our approach}}

{\begin{table*}[htb]
\caption{Experimental results.}
\label{TABLE:NEW}
\begin{center}
{\small\begin{tabular}{c|c|c|c|c|c|c}
 \hline
 \multirow{2}{*}{\textbf{Program} \textbf{(\#Procs)}} & \multirow{2}{*}{\textbf{T}}& \multirow{2}{*}{\textbf{\!Deadlock\!}} &\multicolumn{2}{c|}{\textbf{Time(s)}}&\multicolumn{2}{c}{\textbf{\#Iterations}}\\
 \cline{4-7}
& & &\baseline &\!\ours\! &\baseline &\!\ours\! \\
 \hline
 \multirow{6}{*}{ \texttt{DTG}\! (5) } \!& $o$ & 0 & $10.12$ & \cellcolor{gray!60}$9.02$ & $3$ & \cellcolor{gray!60}$1$ \\
 \cline{2-7}
 & $\textit{m}_1$& 0 & $13.69$ & \cellcolor{gray!60}$9.50$ & $10$ & \cellcolor{gray!60}$2$\\
 \cline{2-7}
 & $\textit{m}_2$& 1 & $10.02$ & \cellcolor{gray!60}$8.93$ & $4$ & \cellcolor{gray!60}$2$ \\
 \cline{2-7}
 & $\textit{m}_3$& 1 & $10.21$ & \cellcolor{gray!60}$9.49$ & $4$ & \cellcolor{gray!60}$2$\\
 \cline{2-7}
 & $\textit{m}_4$& 1 & $10.08$ & \cellcolor{gray!60}$9.19$ & $4$ &\cellcolor{gray!60}$2$ \\
 \cline{2-7}
 & $\textit{m}_5$& 1 & $9.04$ & $9.29$ & $2$ &$2$ \\
 \hline
 \texttt{Matmat}$^{*}$\!(4) & o & 0 & $36.94$ & \cellcolor{gray!60}$10.43$ & $54$ & $\cellcolor{gray!60}$1$$ \\
 \hline
 \multirow{3}{*}{  \texttt{Integrate}\! (6/8/10) } & $o$ & 0/0/0 & $78.17/\textsc{to}/\textsc{to}$ &$\cellcolor{gray!60}8.87/10.45/44.00$ &\!\!$120/3912/3162$\!\! & \cellcolor{gray!60}$1/1/1$ \\
 \cline{2-7}
 & $\textit{m}_1$& 0/0/-1 &$\textsc{to}/\textsc{to}/\textsc{to}$ &$\highlight{49.94}/\textsc{to}/\textsc{to}$ &$4773/3712/3206$ &$\highlight{32}/128/79$ \\
 \cline{2-7}
 & $\textit{m}_2$&\! 1/1/1\! &$9.35/9.83/9.94$ &$9.39/10.76/44.09$ &$\!2/2/2$\! &$2/2/2$ \\
 \hline
 \texttt{Integrate$^{*}$}\! (4/6) & o & 0/0 &  $24.18/123.55$ & \cellcolor{gray!60} $9.39/32.03$ & $27/125$ & \cellcolor{gray!60}$1/1$\\
 \hline
  \multirow{6}{*}{  \texttt{Diffusion2d}\! (4/6)} & $o$ & 0/0 & $106.86/\textsc{to}$ &\cellcolor{gray!60}9.84/13.19 &90/2041& \cellcolor{gray!60}1/1\\
 \cline{2-7}
 & $\textit{m}_1$& 0/1 & $110.25/11.95$  & $\highlight{10.18}/13.81$  & $90/2$  & $\highlight{1}/2$  \\
 \cline{2-7}
 & $\textit{m}_2$& 0/1 &  $3236.02/12.66$ & $\highlight{17.05}/14.38$  & $5850/3$  & \cellcolor{gray!60}$16/2$  \\
 \cline{2-7}
 & $\textit{m}_3$& 0/0 &$\textsc{to}/\textsc{to}$ &$\cellcolor{gray!60}19.26/199.95$ &$5590/4923$ & \cellcolor{gray!60}$16/64$\\
 \cline{2-7}
 & $\textit{m}_4$& 1/1 &$11.35\,/11.52$ &$\highlight{11.14}/14.22$ &$3/2$ &$\highlight{2}/2$ \\
 \cline{2-7}
 & $\textit{m}_5$& 1/0 &$10.98/\textsc{to}$ &\cellcolor{gray!60}$10.85/13.44$ &$2/1991$ &$2/\highlight{1}$ \\
 \hline
 \multirow{2}{*}{ \texttt{Gauss\_elim}\! (6/8/10) } & $o$ & 0/0/0 & $\textsc{to}/\textsc{to}/\textsc{to}$ &\cellcolor{gray!60}13.47/15.12/87.45 & 2756/2055/1662 &\cellcolor{gray!60}1/1/1 \\
 \cline{2-7}
 & $\textit{m}_1$& 1/1/1 & $155.40/\textsc{to}/\textsc{to}$&$\cellcolor{gray!60}14.31/16.99/88.79$ &$121/2131/559$ &\cellcolor{gray!60}$2/2/2$ \\
 \hline
 \multirow{6}{*}{\texttt{Heat}\! (6/8/10)} & $o$ & 1/1/1 & $17.31/17.99/20.51$ &$\highlight{16.75}/19.27/22.75$ & $2/2/2$ & \cellcolor{gray!60}$1/1/1$ \\
 \cline{2-7}
 & $\textit{m}_1$&  1/1/1  &$17.33/18.21/20.78$ &$\highlight{17.03}/19.75/23.16$ &$2/2/2$ &\cellcolor{gray!60}$1/1/1$ \\
 \cline{2-7}
 & $\textit{m}_2$&  1/1/1  &$18.35/18.19/20.74$ &$\highlight{16.36}/19.53/23.07$ &$2/2/2$ &\cellcolor{gray!60}$1/1/1$ \\
 \cline{2-7}
 & $\textit{m}_3$&  1/1/1  &$19.64/20.21/23.08$ & \cellcolor{gray!60}$16.36/19.72/22.95$ &$3/3/3$  &\cellcolor{gray!60}$1/1/1$  \\
 \cline{2-7}
 & $\textit{m}_4$&  1/1/1  &$22.9/24.73/27.78$ &\cellcolor{gray!60}$16.4/19.69/22.90$  &$9/9/9$  &\cellcolor{gray!60}$1/1/1$  \\
 \cline{2-7}
 & $\textit{m}_5$&  1/1/1  &$24.28/28.57/32.67$  &\cellcolor{gray!60}$16.61/19.59/22.42$  &$7/7/7$  &\cellcolor{gray!60}$1/1/1$  \\
 \hline
 \multirow{4}{*}{ \texttt{Mandelbrot}\! (6/8/10) } & $o$ & \!0/0/-1\! & $\textsc{to}/\textsc{to}/\textsc{to}$ &$\highlight{117.68}/\highlight{831.87}/\textsc{to}$ &$500/491/447$ &$\highlight{9}/\highlight{9}/9$ \\
 \cline{2-7}
 & $\textit{m}_1$& \!-1/-1/-1\! &$\textsc{to}/\textsc{to}/\textsc{to}$ &$\textsc{to}/\textsc{to}/\textsc{to}$ &$1037/1621/1459$ &$173/227/246$\\
 \cline{2-7}
 & $\textit{m}_2$& \!-1/-1/-1\! &$\textsc{to}/\textsc{to}/\textsc{to}$ &$\textsc{to}/\textsc{to}/\textsc{to}$&$1093/1032/916$ &$178/136/90$\! \\
 \cline{2-7}
 & $\textit{m}_3$& 1/1/1 &$10.71/11.17/11.92$ &$10.84/11.68/13.5$ &$2/2/2$ &$2/2/2$ \\
  \hline
 \texttt{Mandelbort$^{*}$}\! (4/6)& o & 0/0 &  $68.09/270.65$ & \cellcolor{gray!60} $12.65/13.21$ & $72/240$ & \cellcolor{gray!60}$2/2$\\
 \hline
 \texttt{Sorting$^{*}$}\! (4/6) & o & 0/0 &  $\textsc{to}/\textsc{to}$ & \cellcolor{gray!60} $19.18/46.19$ & $584/519$ & \cellcolor{gray!60}$1/1$\\
 \hline
 \multirow{2}{*}{  \texttt{Image\_mani}\! (6/8/10) }&$o$ & 0/0/0 & \!\!$97.69/118.72/141.87$\!\! &\cellcolor{gray!60}$18.68/23.84/27.89$ &$96/96/96 $ &\cellcolor{gray!60}$4/4/4$ \\
 \cline{2-7}
 & $\textit{m}_1$& 1/1/1 &$12.92/15.80/15.59$ &$14.15/\highlight{14.53}/16.86$ &$2/2/2$ &$2/2/2$ \\
 \hline
 \texttt{DepSolver}\! (6/8/10) & $o$ & 0/0/0 &\!\!$94.17/116.5/148.38$\!\! &\!$97.19/123.36/151.83$\! &$4/4/4$ &$4/4/4$ \\
 \hline
 \multirow{4}{*}{  \texttt{Kfray}\! (6/8/10) } & $o$ & 0/0/0 & $\textsc{to}/\textsc{to}/\textsc{to}$ &\cellcolor{gray!60}$51.59/68.25/226.96$ &$1054/981/1146$ &\cellcolor{gray!60}$1/1/1$ \\
 \cline{2-7}
 & $\textit{m}_1$&  1/1/1 & $52.15/53.50/46.83$ &$53.14/69.58/229.97$ &$2/2/2$ &$2/2/2$ \\
 \cline{2-7}
 & $\textit{m}_2$& \!-1/-1/-1\! & $\textsc{to}/\textsc{to}/\textsc{to}$ &$\textsc{to}/\textsc{to}/\textsc{to}$ &\!\!$1603/1583/1374$\!\!&$239/137/21$ \\
 \cline{2-7}
 & $\textit{m}_3$& 1/1/1 & $51.31/43.34/48.33$ &$\highlight{50.40}/71.15/230.18$ &$2/2/2$ &$2/2/2$ \\
 \hline
 \texttt{Kfray$^{*}$}\! (4/6) & o & 0/0 &  $\textsc{to}/\textsc{to}$ & \cellcolor{gray!60} $53.44/282.46$ & $1301/1575$ & \cellcolor{gray!60}$1/1$\\
 \hline
 \multirow{6}{*}{\texttt{Clustalw}\! (6/8/10) } & $o$ & 0/0/0 & $\textsc{to}/\textsc{to}/\textsc{to}$ & \cellcolor{gray!60}$47.28 / 79.38 / 238.37$ &1234/1105/1162 &\cellcolor{gray!60}1/1/1 \\
 \cline{2-7}
 & $\textit{m}_1$& 0/0/0 &$\textsc{to}/\textsc{to}/\textsc{to}$ &\cellcolor{gray!60}$47.94 / 80.10 / 266.16$ &$1365/1127/982$ &\cellcolor{gray!60}$1/1/1$ \\
 \cline{2-7}
 & $\textit{m}_2$& 0/0/0 &$\textsc{to}/\textsc{to}/\textsc{to}$ &\cellcolor{gray!60}$47.71 / 90.32 / 266.08$ &$1241/1223/915$ &\cellcolor{gray!60}$1/1/1$ \\
 \cline{2-7}
 & $\textit{m}_3$& 1/1/1 &$895.63/\textsc{to}/\textsc{to}$ &\cellcolor{gray!60}$\!149.71 / 1083.95 / 301.99\!$ &$175/1342/866$ &\cellcolor{gray!60}$5/17/2$ \\
 \cline{2-7}
 & $\textit{m}_4$& 0/0/0 &$\textsc{to}/\textsc{to}/\textsc{to}$ &\cellcolor{gray!60}$47.49 / 79.94 / 234.99$ &$1347/1452/993$ &\cellcolor{gray!60}$1/1/1$ \\
 \cline{2-7}
 & $\textit{m}_5$& 0/0/0 &$\textsc{to}/\textsc{to}/\textsc{to}$ &\cellcolor{gray!60}$47.75 / 80.33 / 223.77$ &$1353/1289/1153$ &\cellcolor{gray!60}$1/1/1$ \\
 \hline
\end{tabular}}
\end{center}
\end{table*}}

\textbf{Baselines%
.}\quad
We use pure symbolic execution as the first baseline because: (1) none of the state-of-the-art symbolic execution based verification tools can analyze
non-blocking MPI programs, \eg, CIVL~\cite{DBLP:conf/pvm/LuoZS17, DBLP:conf/sc/SiegelZLZMEDR15};
(2) MPI-SPIN~\cite{DBLP:conf/pvm/Siegel07} can support input coverage and non-blocking operations, but it requires building models of the programs manually;
and (3) other automated tools that support non-blocking operations, such as
MOPPER~\cite{forejt2014precise} and ISP~\cite{DBLP:conf/cav/VakkalankaGK08}, can only verify programs under given inputs. MPI-SV aims at covering both the input space and non-determinism automatically.
To compare with pure symbolic execution, we run \mpise\ under two configurations: (1) \baseline, \emph{i.e.}, applying only symbolic execution for path exploration, and (2) \ours, \emph{i.e.}, using model checking based boosting.
Most of the programs run with 6, 8, and 10 processes, respectively. $\texttt{DTG}$ and $\texttt{Matmat}$ can only be run with 5 and 4 processes, respectively. For $\texttt{Diffusion}$ and the programs using the  \emph{master-slave} pattern, we only run them with 4 and 6 processes due to the huge path space. We use \mpise\ to verify deadlock freedom of MPI programs and also evaluate 2 \emph{non-reachability} properties for \texttt{Integrate} and \texttt{Mandelbrot}.
The timeout is one hour. There are three possible verification results: finding a violation, no violation, or timeout. We carry out all the tasks on an Intel Xeon-based Server with 64G memory and 8 2.5GHz cores running a Ubuntu 14.04 OS. We ran each verification task three times and use the average results to alleviate the experimental errors.
To evaluate \mpise's effectiveness further, we also directly compare \mpise with
CIVL~\cite{DBLP:conf/pvm/LuoZS17, DBLP:conf/sc/SiegelZLZMEDR15} and MPI-SPIN~\cite{DBLP:conf/pvm/Siegel07}. Note that, since MPI-SPIN needs manual modeling, we only use \mpise\ to verify MPI-SPIN's C benchmarks \wrt deadlock freedom.

\subsection{Experimental Results}\label{exp-result}

Table \ref{TABLE:NEW} lists the results for evaluating \mpise\ against pure symbolic execution. The first column shows program names, and \textbf{\#Procs} is the number of running processes. \textbf{T} specifies whether the analyzed program is mutated, where $o$ denotes the original program, and $\textit{m}_i$ represents a mutant. A task comprises a program and the number of running processes. We label the programs using \emph{master-slave} pattern with superscript ``*''.
Column \textbf{Deadlock} indicates whether a task is deadlock free, where 0, 1, and -1 denote \emph{no deadlock}, \emph{deadlock} and \emph{unknown}, respectively. We use unknown for the case that both configurations fail to complete the task. Columns \textbf{Time(s)} and \textbf{\#Iterations} show the verification time and the number of explored paths, respectively, where \textsc{to} stands for timeout. The results where \ours\ performs better is in gray background.

For the 111 verification tasks, \mpise\ completes 100 tasks ($90\%$) within one hour, whereas 61 tasks ($55\%$) for \baseline. \ours\ detects deadlocks in 48 tasks, while the number of \baseline\ is 44. We manually confirmed that the detected deadlocks are real.
For the 48 tasks having deadlocks, \mpise\ on average offers a 5x speedups for detecting deadlocks.
On the other hand, \ours\ can verify deadlock freedom for 52 tasks, while only 17 tasks for \baseline. \mpise\ achieves an average 19x  speedups. Besides, compared with \baseline, \ours\ requires fewer paths to detect the deadlocks (1/55 on average) and complete the path exploration (1/205 on average). These results demonstrate \mpise's effectiveness and efficiency.

\begin{figure}[!t]
\centering
\begin{tikzpicture}
\pgfplotscreateplotcyclelist{my black white}{%
solid, every mark/.append style={solid, fill=gray}, mark=diamond*\\%
densely dashed, every mark/.append style={solid, fill=gray}, mark=square*\\%
}
\begin{axis}[
   width=5cm,
   height=4cm,
   scale only axis,
   xmin=0, xmax=60,
   xtick = {0,5,10,15,20,25,30,35,40,45,50,55,60},
   ymin=0, ymax=111,
   axis lines*=left,
   yticklabel style={font=\tiny},
   xticklabel style={font=\tiny},
   legend style ={ at={(0.45, 0.24)},
   anchor=north west, draw=black,
   fill=white,align=left,font=\tiny},
   cycle list name=my black white,
   smooth,
   xlabel={\footnotesize{Verification time thresholds}},
   ylabel={\footnotesize{\# Completed verification tasks}}
]
    \addplot coordinates{
(0,0)
(5,59)
(10,59)
(15,60)
(20,60)
(25,60)
(30,60)
(35,60)
(40,60)
(45,60)
(50,60)
(55,61)
(60,61)   };
   \addlegendentry{\baseline};
    \addplot coordinates{
(0,0)
(5,96)
(10,97)
(15,98)
(20,99)
(25,99)
(30,99)
(35,100)
(40,100)
(45,100)
(50,100)
(55,100)
(60,100)   };
   \addlegendentry{\ours};
   \end{axis}
\end{tikzpicture}%
\caption{Completed tasks under a time threshold.}\label{Trend}
\end{figure}
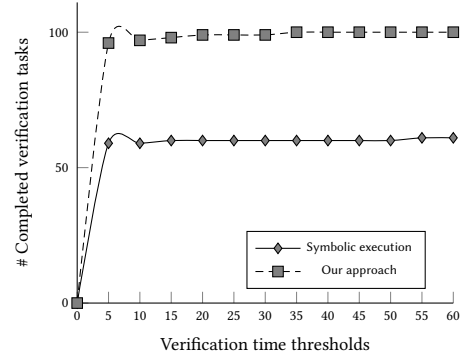

Figure \ref{Trend} shows the efficiency of verification for the two configurations. The X-axis varies the time threshold from 5 minutes to one hour, while the Y-axis is the number of completed verification tasks. \ours\ can complete more tasks than \baseline\ under the same time threshold, demonstrating \mpise's efficiency.
In addition, \mbox{\ours} can complete $96\ (96\%)$ tasks in $5$ minutes, which also demonstrates \mpise's effectiveness.

For some tasks, \eg, \texttt{Kfray}, \mpise does not outperform \baseline. The reasons include: (a) the paths contain hundreds of non-wildcard operations, and the corresponding CSP models are huge, and thus time-consuming to model check; (b) the number of wildcard receives or their possible matchings is very small, and as a result, only few paths are pruned.

\noindent
\quad\textbf{\emph{Comparison with CIVL}.}
CIVL uses symbolic execution to build a model for the whole program and performs model checking on the model. In contrast, \mpise\ adopts symbolic execution to generate \emph{path-level verifiable} models. CIVL does not support non-blocking operations.
We applied CIVL on our evaluation subjects. It only successfully
analyzed \texttt{DTG}. \texttt{Diffusion2d} could be analyzed after removing
unsupported external calls. MPI-SV and CIVL had similar performance on
these two programs. CIVL failed on all the remaining programs due to
compilation failures or lack of support for non-blocking operations. In contrast, \mpise\ successfully analyzed 99 of the 140 programs in CIVL's latest benchmarks. The failed ones are small API test programs for the APIs that
MPI-SV does not support. For the real-world program \texttt{floyd} that both MPI-SV and CIVL can analyze, MPI-SV verified its deadlock-freedom under 4 processes in 3 minutes, while CIVL timed out after 30 minutes. The results indicate the benefits of \mpise's\ path-level modeling.

\noindent
\quad\textbf{\emph{Comparison with MPI-SPIN}.} MPI-SPIN relies on manual modeling of MPI programs. Inconsistencies may happen between an MPI program and its model. Although prototypes exist for translating C to Promela~\cite{jiang2009using}, they are impractical for real-world MPI programs. MPI-SPIN's state space reduction treats communication channels as rendezvous ones; thus, the reduction cannot handle the programs with wildcard receives. MPI-SV leverages model checking to prune redundant paths caused by wildcard receives.
We applied \mpise\ on MPI-SPIN's 17 C benchmarks to verify deadlock freedom, and \mpise\ successfully analyzed 15 automatically, indicating the effectiveness. For the remaining two programs, \ie, \texttt{BlobFlow} and \texttt{Monte}, MPI-SV cannot analyze them due to the lack of support for APIs. For the real-world program \texttt{gausselim}, MPI-SPIN needs 171s to verify that the model is deadlock-free under 5 processes, while MPI-SV only needs 27s to verify the program automatically. If the number of the processes is 8, MPI-SPIN timed out in 30 minutes, but MPI-SV used 66s to complete verification.

\noindent
\quad\textbf{\emph{Temporal properties}.}
We specify two temporal safety properties $\varphi_1$ and $\varphi_2$ for \texttt{Integrate} and \texttt{Mandelbrot}, respectively, where $\varphi_1$ requires process one cannot receive a message before process two, and $\varphi_2$ requires process one cannot send a message before process two. Both $\varphi_1$ and $\varphi_2$ can be represented by an LTL formula
$!a$ \textbf{U} $b$, which requires event $a$ cannot happen before event $b$. We verify \texttt{Integrate} and \texttt{Mandelbrot} under 6 processes. The verification results show that \mpise\ detects the violations of $\varphi_1$ and $\varphi_2$, while pure symbolic execution fails to detect violations.

\noindent
\quad\textbf{\emph{Runtime bugs}.} \mpise can also detect local runtime bugs. During the experiments, \mpise finds 5 \emph{unknown} memory access out-of-bound bugs: 4 in \texttt{DepSolver} and 1 in \texttt{ClustalW}.

\section{Related Work}\label{relatedwork}

Dynamic analyses are widely used for analyzing MPI programs.
Debugging or testing tools \citep{TotalView, DDT, krammer2004marmot, DBLP:conf/parco/SamofalovKKZKD05,hilbrich2012mpi,DBLP:journals/cacm/LagunaASGLSBKZC15,DBLP:conf/pldi/MitraLABSG14} have better feasibility and scalability but depend on specific inputs and running schedules.
Dynamic verification techniques, \eg, ISP~\citep{DBLP:conf/cav/VakkalankaGK08} and DAMPI~\citep{vo2010scalable}, run MPI programs multiple times to cover the schedules under the same inputs. \citet{DBLP:conf/fm/BohmMJ16} propose a state-space reduction framework for the MPI program with \emph{non-deterministic synchronization}. These approaches can detect the bugs depending on specific matchings of wildcard operations, but may still miss inputs related bugs. \mpise\ supports both input and schedule coverages, and a larger scope of verifiable properties. MOPPER~\citep{forejt2014precise} encodes the deadlock detection problem under concrete inputs in a SAT equation. Similarly, 
\citet{DBLP:conf/nfm/HuangM15} use an SMT formula to reason about a trace of an MPI program for deadlock detection. However, the SMT encoding is specific for the zero-buffer mode. 
\citet{Hermes} combines dynamic and symbolic analyses to verify \emph{multi-path} MPI programs. Compared with these path reasoning work in dynamic verification, \mpise\ ensures input space coverage and can verify more properties, \ie, safety and liveness properties in LTL. Besides, \mpise\ employs CSP to enable a more expressive modeling, \eg, supporting conditional completes-before~\citep{DBLP:conf/cav/VakkalankaGK08} and master-slave pattern~\citep{Gropp:2014:UMP:2717061}.

For static methods of analyzing MPI program, \mbox{MPI-SPIN}~\citep{DBLP:conf/pvm/Siegel07,siegel:2007:vmcai} manually models MPI programs in Promela~\citep{iosif2009promela}, and verifies the model \wrt LTL properties~\citep{DBLP:books/daglib/0077033} by SPIN~\citep{holzmann1997the} (\cf Section \ref{exp-result} for empirical comparison). MPI-SPIN can also verify the consistency between an MPI program and a sequential program, which is not supported by \mpise.
\citet{DBLP:conf/cgo/Bronevetsky09} proposes parallel control flow graph (pCFG) for MPI programs to capture the interactions between arbitrary processes. But the static analysis using pCFG is hard to be automated.
ParTypes \citep{DBLP:conf/oopsla/LopezMMNSVY15} uses type checking and deductive verification to verify MPI programs against a protocol. ParTypes's verification results are sound but incomplete, and independent with the number of processes. ParTypes does not support non-deterministic or non-blocking MPI operations.
MPI-Checker~\cite{DBLP:conf/sc/DrosteKL15} is a static analysis tool built on Clang Static Analyzer~\cite{CSA}, and only supports intraprocedural analysis of local properties such as double non-blocking and missing wait. \citet{DBLP:conf/vmcai/BotbolCG17} abstract an MPI program to symbolic transducers, and obtain the reachability set based on abstract interpretation~\cite{DBLP:conf/popl/CousotC77}, which only supports blocking MPI programs and may generate false positives. 
COMPI \cite{DBLP:conf/ipps/LiLBC018, Li:2019:ECT:3302516.3307353} uses concolic testing \cite{godefroid2005dart,DBLP:conf/sigsoft/SenMA05} to detect assertion or runtime errors in MPI applications. \citet{DBLP:conf/sc/YeZS18} employs partial symbolic execution \cite{DBLP:conf/uss/RamosE15} to detect MPI usage anomalies. However, these two symbolic execution-based bug detection methods do not support the non-determinism caused by wildcard operations. \citet{DBLP:conf/sc/LuoS18} propose a preliminary deductive method for verifying the numeric properties of MPI programs in an unbounded number of processes. However, this method still needs manually provided verification conditions to prove MPI programs.


\mpise\ is related to the existing work on symbolic execution~\citep{king1976symbolic}, which has been advanced significantly during the last decade~\citep{godefroid2005dart, DBLP:conf/sigsoft/SenMA05, cadar2008klee, bucur2011parallel, DBLP:conf/ndss/GodefroidLM08, DBLP:conf/issta/PasareanuMBGLPP08, DBLP:conf/tap/TillmannH08, DBLP:conf/icse/Wang0CZWL18, zhang2015regular}. Many methods have been proposed to prune paths during symbolic execution~\citep{boonstoppel2008rwset, jaffar2013boosting, DBLP:conf/asplos/CuiHWY13, guo2015assertion, DBLP:conf/icse/YuCWS018}. The basic idea is to use the techniques such as slicing~\citep{DBLP:conf/pldi/JhalaM05} and interpolation~\citep{DBLP:conf/tacas/McMillan05} to safely prune the paths. Compared with them, \mpise\ only prunes the paths of the same path constraint but different message matchings or operation interleavings. 
\mpise\ is also related to the work of automatically extracting session types \cite{DBLP:conf/cc/NgY16}  or behavioral types \cite{DBLP:conf/icse/LangeNTY18} for Go programs and verifying the extracted type models. These methods extract over-approximation models from Go programs, and hence are sound but incomplete. Compared with them, \mpise\ extracts path-level models for verification. Furthermore, there exists work of combining symbolic execution and model checking~\citep{nori2009yogi, su2015combining, daca2016abstraction}. YOGI~\citep{nori2009yogi} and Abstraction-driven concolic testing~\citep{daca2016abstraction} combine dynamic symbolic execution ~\citep{godefroid2005dart,DBLP:conf/sigsoft/SenMA05} with counterexample-guided abstraction refinement (CEGAR)~\citep{clarke2000counterexample}. \mpise\ focuses on parallel programs, and the verified models are path-level. \mpise is also related to the work of unbounded verification for parallel programs \citep{DBLP:conf/popl/BouajjaniE12,DBLP:conf/cav/BouajjaniEJQ18,DBLP:journals/pacmpl/BakstGKJ17,DBLP:journals/pacmpl/GleissenthallKB19}. Compared with them, \mpise is a bounded verification tool and supports the verification of LTL properties. Besides, \mpise\ is related to the existing work of testing and verification of shared-memory programs \citep{DBLP:conf/osdi/MusuvathiQBBNN08,guo2015assertion,DBLP:conf/sigsoft/GuoWW18,DBLP:conf/oopsla/DemskyL15,DBLP:conf/pldi/HuangZD13,DBLP:conf/oopsla/Huang016,DBLP:journals/fmsd/ChakiCGOSY04,DBLP:conf/tacas/CimattiNR11,DBLP:conf/concur/KraglQH18,DBLP:conf/kbse/InversoN0TP15,DBLP:conf/tacas/Yin0LLW18}. Compared with them, \mpise concentrates on message-passing programs. Utilizing the ideas in these work for analyzing MPI programs is interesting and left to the future work.

\section{Conclusion}\label{conclusion}

We have presented \mpise\ for verifying MPI programs with both non-blocking and non-deterministic operations. By synergistically combining symbolic execution and model checking, \mpise\ provides a general framework for verifying MPI programs. We have implemented \mpise\ and extensively evaluated it on real-world MPI programs. The experimental results are promising demonstrate \mpise's effectiveness and efficiency. The future work lies in several directions: (1) enhance MPI-SV to support more MPI operations, (2) investigate the automated performance tuning of MPI programs based on MPI-SV, (3) apply our synergistic framework to other message-passing programs.

%


\bibliography{main}

\ifdefined\supplementary 
\appendix
\section{Appendix}
\subsection{Semantics of the Core MPI Language}\label{mpi_semantics}

\textbf{Auxiliary Definitions.}\quad Before giving the MPI language's semantics, we give some auxiliary definitions. Given an MPI program $\mathcal{MP} = \{\textsf{Proc}_i \mid 0 \le i \le n \}$, $\mathit{send(dst)}$ and $\mathit{recv(src)}$ denote $\mathcal{MP}$'s send and receive operations\footnote{$\mathit{send(dst)}$ and $\mathit{recv(src)}$ can denote both blocking and unblocking operations, and we omit the $req$ parameter for non-blocking ones for the sake of simplicity.}, respectively, where $dst{\in}\{0,\dots,n\}$ and $src{\in}\{0,\dots,n\}{\cup}\{*\}$. $\mathsf{op}(\mathcal{MP})$ represents the set of all the MPI operations in $\mathcal{MP}$, $\mathit{rank}(\alpha)$ is the process identifier of operation $\alpha$, and $\mathit{isBlocking}(\alpha)$ indicates whether $\alpha$ is a blocking operation.

\begin{figure*}
$\frac{s_i.\mathcal{F}=\mathsf{active}}
{ S\xrightarrow{\issue{s_i.Stat}}(\dots,s_i[update(\mathcal{F},s_i.Stat), \mathcal{B}.push(s_i.Stat)],\dots)}${\small$\langle \mathsf{ISSUE}\rangle$}\hspace{2mm}
$\frac{\forall i\in [0,n],s_i.\mathcal{F}=\mathsf{blocked}\land (\exists \alpha \in s_i.\mathcal{B}, \alpha =\texttt{Barrier})}
{ S\xrightarrow{B}(s_0[update(\mathcal{F},\alpha),\mathcal{B}.pull(\alpha)], ..., s_n[update(\mathcal{F},\alpha),\mathcal{B}.pull(\alpha)])}${\small$\langle \mathsf{B}\rangle$}

\vspace{4mm}

$\frac{\exists\alpha\in s_i.\mathcal{B},\exists\beta\in s_j.\mathcal{B}, ready(\alpha,s_i)\land ready(\beta,s_j)\land match(\alpha,\beta)\land C(\alpha,s_i,\beta,s_j) }
{ S\xrightarrow{SR/SR^*}(\dots,
s_i[update(\mathcal{F},\alpha),\mathcal{B}.pull(\alpha)],\dots,s_j[update(\mathcal{F},\beta),\mathcal{B}.pull(\beta)],\dots)}${\small$\langle \mathsf{SR}\rangle$}
$\quad\quad$\hspace{2mm}
$\frac{s_i.\mathcal{F}=\mathsf{blocked}\land \exists\alpha\in s_i.\mathcal{B},(\alpha=\texttt{Wait}\land ready(\alpha,s_i))}
{ S\xrightarrow{W}(\dots,s_i[update(\mathcal{F},\alpha),\mathcal{B}.pull(\alpha)],\dots)}${\small$\langle \mathsf{W}\rangle$}
\caption{Transition Rules of MPI operations}\label{rules}
\end{figure*}

\begin{defn}\label{procee-def}
\textbf{MPI Process State.} An MPI process's state
is a tuple $(\mathcal{M}, \mathit{Stat}, \mathcal{F},\mathcal{B}, \mathcal{R})$, where $\mathcal{M}$ maps a variable to its value,
$\mathit{Stat}$ is the next program statement to execute,
$\mathcal{F}$ is the flag of process status and belongs to the set \{$\mathsf{active}$, $\mathsf{blocked}$, $\mathsf{terminated}$\}, $\mathcal{B}$ and $\mathcal{R}$ are infinite buffers to store the issued MPI operations not yet matched and the matched MPI operations, respectively.
\end{defn}

An element $elem$ of a process state $s$ can be accessed by $s.elem$, \eg, $s.\mathcal{F}$ is the status of $s$. The behavior of a process can be regarded as a sequence of statements, and we use $index(\alpha)$ to denote the index of operation $\alpha$ in the sequence. An MPI program's global state $S$ is composed by the states of the MPI processes, \ie, $S=(s_0,\dots,s_{n})$. An MPI program's semantics is a labeled transition system defined below.

\begin{defn}
\textbf{Labeled Transition System.} A labeled transition system (LTS) of an MPI program $\mathcal{MP}$ is a quadruple $(\mathcal{G},\Sigma,\rightarrow,\mathcal{G}_0)$, where $\mathcal{G}$ is the set of global states, $\Sigma$ denotes the set of actions defined below, $\rightarrow \subseteq \mathcal{G}\times\Sigma\times \mathcal{G}$ represents the set of transitions, and $\mathcal{G}_0$ is the set of initial states.
\end{defn}

\textbf{Actions.}\quad The action set $\Sigma$ is $\{B, W, \mathit{SR}, \mathit{SR}^*\}\cup\{\issue{o} \mid o\in \mathsf{op}(\mathcal{MP})\}$, where $B$ represents the synchronization of \verb"Barrier", $W$ is the execution of \verb"Wait" operation, $SR$ denotes the matching of message send and deterministic receive, $SR^*$ represents the matching of message send and wildcard receive, and $\issue{o}$ stands for the issue of operation $o$.

\textbf{Transition Rules.}\quad We first give some definitions used by the transition rules. We use $ready(\alpha,s_i) \in \{\mathsf{True}, \mathsf{False}\}$ defined as follows to indicate whether operation $\alpha$ is ready to be matched in state $s_i$ \wrt the MPI standard \cite{MPI}, where $\beta \in s_i.\mathcal{B}$ represents that operation $\beta$ is in the buffer $\mathcal{B}$ of process state $s_i$ and $k\in\{0,\dots,n\}$.
\begin{itemize}[leftmargin=2em]
\item
If $\alpha$ is $\verb"Wait(r)"$, $\alpha$ can be matched if the waited operation has been matched, \ie, $\exists\beta{\in}\mathcal{R},(\beta{=}$ $\verb"ISend(k,r)"\vee\beta {=} \verb"IRecv(k,r)"\vee\beta {=} \verb"IRecv(*,r)")$.
\item
If $\alpha$ is $send(k)$, $\alpha$ can be matched if there is no previously issued $send(k)$ not yet matched, \ie, $\neg(\exists \beta{\in} s_i.\mathcal{B},index(\beta)$\\$ < index(\alpha)\land\beta=send(k))$.
\item
If $\alpha$ is $recv(k)$, $\alpha$ can be matched if the previously issued $recv(k)$ or $recv(*)$ has been matched, which can be formalized as $\neg(\exists\beta{\in} s_i.\mathcal{B}, \ index(\beta) <index(\alpha)\land(\beta=recv(k)\!\lor\beta = recv(*)))$.
\item
If $\alpha$ is $recv(*)$, $\alpha$ can be matched if
the previously issued $recv(*)$ has been matched, \ie, $\neg(\exists\beta{\in} s_i.\mathcal{B},\ index(\beta) <index(\alpha)\land\beta=recv(*))$. It is worth noting the \textit{conditional completes-before} pattern~\cite{DBLP:conf/cav/VakkalankaGK08}, \ie, operation $\verb"IRecv(k,r)"$ followed by a $recv(*)$, and the $recv(*)$ can complete first when the matched message is not from $k$. We will give a condition later to ensure such relation.
\end{itemize}

Suppose $\alpha$ and $\beta$ are MPI operations, $match(\alpha,\beta)$ represents whether $\alpha$ and $\beta$ can be matched statically, and can be defined as  $match^{'}(\alpha,\beta)\lor match^{'}(\beta,\alpha)$, where $match^{'}(\alpha,\beta)$ is
$((\alpha,\beta){=}(send(dst),$\\$recv(src))\land(dst{=}src\lor src=*)$. We use $s_i[ops]$ to denote the updates of the process state $s_i$ with an update operation sequence $ops$. The operation $update(\mathcal{F},\alpha)$ updates the process status \wrt $\alpha$ as follows.
$$
update(\mathcal{F}\!,\!\alpha)\!\!=\!\!
\left\{
\begin{array}{ll}
\!\!\mathcal{F}{:=}\mathsf{active} \!\!&\!\! \mathcal{F}{=}\mathsf{blocked} \land \mathit{isBlocking}(\alpha) \\
\!\!\mathcal{F}{:=}\mathsf{blocked} \!\!&\!\! \mathcal{F}{=}\mathsf{active} \land \mathit{isBlocking}(\alpha) \\
\!\!\mathcal{F}{:=}\mathcal{F}\!\! &\!\! \mathit{otherwise}
\end{array}\right.
$$
$\mathcal{B}.push(\alpha)$ represents adding MPI operation $\alpha$ to buffer $\mathcal{B}$, while $\mathcal{B}.pull(\alpha)$ represents removing $\alpha$ from $\mathcal{B}$ and adding $\alpha$ to $\mathcal{R}$.
We use $Stat'$ to denote the statement next to $Stat$. We use $C(\alpha,s_i,\beta,s_j)$ = $C_1(\alpha,s_i,\beta,s_j) \wedge C_1(\beta,s_j,\alpha,s_i)$ to define the conditional completes-before relation requirement, where $C_1(\alpha,s_i,\beta,s_j)$ is $\neg(\exists\beta' \in s_j.\mathcal{B},$\\$\ (\beta{=}\verb"IRecv(*,r)"\land\beta'{=}\verb"IRecv(i,r')"\land ready(\beta',s_j)\land match(\alpha,\beta')))$.

Figure \ref{rules} shows four transition rules for MPI operations. For the sake of brevity, we omit the transition rules of local statements, which only update the mapping $\mathcal{M}$ and the next statement $Stat$ to execute. Rule $\langle \mathsf{ISSUE}\rangle$ describes the transition of issuing an MPI operation, which requires the issuing process to be $\mathsf{active}$. After issuing the operation, the process status is updated, the next statement to execute becomes $Stat'$ (omitted for the sake of spaces), and the issued operation is added to the buffer $\mathcal{B}$. Rule $\langle \mathsf{SR}\rangle$ is about matchings of message send and receive. There are three required conditions to match a send to a receive: (1) both of them have been issued to the buffer $\mathcal{B}$ and are ready to be matched; (2) operation arguments are matched, \ie, $match(\alpha,\beta)$; (3) they comply with the conditional completes-before relation, \ie, $C(\alpha,s_i,\beta,s_j)$. After matching, the matched operations will be removed from buffer $\mathcal{B}$ and added to buffer $\mathcal{R}$, and the process status is updated. Rule $\langle \mathsf{B}\rangle$ is for barrier synchronization, which requires that all the processes have been blocked at the \verb"Barrier". After barrier synchronization, operation \verb"Barrier" will be moved from buffer $\mathcal{B}$ to buffer $\mathcal{R}$ and all the processes become $\mathsf{active}$. Rule $\langle \mathsf{W}\rangle$ is for \verb"Wait" operation, which requires the corresponding non-blocking operation has been finished. After executing \verb"Wait" operation, the process becomes $\mathsf{active}$ and the \verb"Wait" operation will be removed from buffer $\mathcal{B}$ and added to buffer $\mathcal{R}$.

\subsection{Correctness of Symbolic Execution for MPI Programs}\label{correctness}

Round-robin schedule and \lazy  symbolic execution is an instance of model checking with POR preserving reachability properties. Next, we prove the correctness of symbolic execution method for verifying reachability properties.

\begin{defn}\label{def-property}\textbf{Reachability Property.} A reachability property $\varphi$ of an MPI program $\mathcal{MP}$ can be defined as follows, where $\mathbf{assertion}(S)$ represents an assertion of global state $S$, \eg, deadlock and assertions of variables.$$
\begin{array}{lll}
\gamma & ::= & \mathbf{true} \mid \gamma \vee \gamma \mid \neg\gamma \mid \mathbf{assertion}(S) \\
\varphi & ::= & \mathbf{EF}\ \gamma \mid \neg\varphi
\end{array}
$$$\mathbf{EF}\ \gamma$ returns true iff there exists an $\mathcal{MP}$'s state that satisfies the formula $\gamma$.
\end{defn}

We use $S\xrightarrow{a}S'$ to represent the transition $(S, a, S')$ in the transition set and $\enable{S}$ to denote the set of enabled actions at global state $S$, \ie, $\enable{S} {=} \{a \mid \exists S'\in \mathcal{G}, S\xrightarrow{a}S'\}$. If $S\xrightarrow{a}S'$, we use $a(S)$ to represent $S'$.
Instead of exploring all the possible states, MPI-SV only executes a subset of $\enable{S}$ (denoted as $E(S)$) when reaching a state $S$. According to the workflow of MPI symbolic execution, we define $E(S)$ below, where
\begin{itemize}[leftmargin=2em]
\item $\mathit{minIssue}(S) = \mathsf{min}\{rank(o) \mid \issue{o}\in \enable{S}\}$ is the minimum process identifier that can issue an MPI operation.
\item $\mathit{minRank}(S) = \mathsf{min}\{rank_a(b) \mid b \in \enable{S} \land b \neq SR^{*} \}$ is the minimum process identifier of enabled non-wildcard actions.
When $b$ is $SR$, $rank_a(b)$ is the process identifier of the send; when $b$ is $W$, $rank_a(b)$ is the process identifier of the \verb"Wait" operation.
\end{itemize}
{\small$$
E(S){=}\left\{\hspace{-2mm}
\begin{array}{ll}
\{\issue{o}\} \!\!\!& if\ \issue{o}{\in} enabled(S) {\land} rank(o) {=}\mathit{minIssue}(S) \\
\{B\} \!\!\!& if\ B{\in} enabled(S) \\
\{W\} \!\!\!& if\ W{\in} enabled(S){\land} rank_a(W){=}\mathit{minRank}(S) \\
\{SR\} \!\!\!& if\ SR{\in} enabled(S){\land} rank_a(SR){=}\mathit{minRank}(S) \\
enabled(S) \!\!\!& otherwise \\
\end{array}\right.
$$}
When $issue(o)$ is enabled in state $S$, we will select the enabled issue operation having the smallest process identifier as $E(S)$, which is in accordance with round-robin schedule. Remember that \lazy matching delays the matching of wildcard receive ($SR^*$) as late as possible. If $B$ is enabled, we use $\{B\}$ as $E(S)$; else, we will use the deterministic matching ($W$ or $SR$) having smallest process identifier (we  select $W$ in case $rank_a(W){=}rank_a(SR)$); otherwise, we use the enabled set as $E(S)$, \ie, $\forall a {\in} \enable{S},\ a{=}SR^*$.

\begin{defn}
\textbf{Independence Relation ($\mathcal{I}\subseteq \Sigma\times\Sigma$).}
For $S\in \mathcal{G}$ and $(a, b)\in \mathcal{I}$, $\mathcal{I}$ is a binary relation that if $(a, b)\in enabled(S)$ then
$a \in enabled(b(S))$, $b \in enabled(a(S))$, and $a(b(S))=b(a(S))$.
\end{defn}

The dependence relation $\mathcal{D}\subseteq\Sigma\times\Sigma$ is the complement of $\mathcal{I}$, \ie, $(a,b)\in\mathcal{D}$, if $(a,b)\notin\mathcal{I}$.
Given an MPI program $\mathcal{MP}$, whose LTS model is $M = (\mathcal{G},\Sigma,\rightarrow,\mathcal{G}_0)$, an execution trace $T = \langle a_0, \dots, a_n \rangle \in \Sigma^*$ of $\mathcal{MP}$ is a sequence of actions, such that $\exists S_i, S_{i+1}\in \mathcal{G}, S_i\xrightarrow{a_i}S_{i+1}$ for each $i\in [0, n]$ and $S_0 \in \mathcal{G}_0$. We use $\mathsf{r}(T)$ to represent the result state of $T$, \ie, $S_{n+1}$, and $T\in M$ to represent $T$ is an execution trace of $\mathcal{MP}$.

\begin{defn}
Execution equivalence $\equiv\;\subseteq\Sigma^*\times\Sigma^*$ is a reflexive and symmetric binary relation  such that (1) $\forall\ a,b\in\Sigma,\ (a,b)\in\mathcal{I}\Rightarrow \langle a, b\rangle\equiv \langle b, a\rangle$. (2) $\forall\ T,T'\in\Sigma^*,\ T\equiv T'$ if there exists a sequence $\langle T_0,\dots,T_k\rangle$ that $T_0{=}T,\ T_k{=}T'$, and for every $i{<}k,\ T_i=u\cdot\langle a, b\rangle \cdot v,\ T_{i+1}{=}u\cdot\langle b, a\rangle \cdot v$, where $(a,b){\in}\mathcal{I}$, $u, v{\in} \Sigma^*$, and $u\cdot v$ is the concatenation of $u$ and $v$.
\end{defn}

Given an MPI program $\mathcal{MP}$, suppose $\mathcal{MP}$'s semantic model is $M$ and the $E(S)$ based model is $M'$, once $E(S)$ satisfies conditions $\mathbf{C1}$ and $\mathbf{C2}$, then for every execution trace $T$, if
$T{\in} M{\land} T{\not\in} M'$, there exists an execution trace $T'$ in $M'$ that $T{\equiv} T'$~\citep{DBLP:journals/fuin/PenczekSGK00}.

$\mathbf{C1.}$\quad $\forall$ $a\in E(S)$, if $(a,b)\in \mathcal{D}$, then for every trace
$S_0 \xrightarrow{a_0} S_1\xrightarrow{a_1},\dots,\xrightarrow{a_k} S_k\xrightarrow{\beta} S_{i+1}$, there exists $a_i\in E(S)$,
where $0\leq i\leq k$.

$\mathbf{C2.}$\quad On every cycle in $M'$, there exists at least one node $S$ that $E(S)=\enable{S}$.

Before proving our selection of $E(S)$ satisfies $\mathbf{C1}$ and $\mathbf{C2}$, we first give a theorem to show the independence relation of co-enabled actions.

\begin{thm}\label{thm:independence}
Given action $a\in\{\issue{o},B,W,SR\}$ and action $b$, for $S\in \mathcal{G}$, if $a,b\in \enable{S}$, then $(a,b)\in\mathcal{I}$.
\end{thm}
\begin{proof}
We only prove the case that $ a$ is $\issue{o}$. The proofs for the other three cases are similar. When $ a{=}\issue{o}$ and $ a, b{\in} \enable{S}$, $ b$ can be $\issue{o'}$, $W$, $SR$, or $SR^*$ \wrt the transition rules.

$(1)\  b{=}\issue{o'}.$
Suppose $rank(o){=}i$, $rank(o'){=}j$, then $i{\neq} j$. Since an issued operation can only block its process, we can have $b \in enabled(a(S))$ and $ a \in enabled(b(S))$. In addition, {\small$a(b(S)) =  b(a(S)) =$\\$(\dots, s_i[update(\mathcal{F},o),Stat{:=}Stat', \mathcal{B}.push(o)],\!\dots,s_j[update(\mathcal{F},\!o'),$}\\{\small$Stat{:=}Stat',\mathcal{B}.push(o')],\!\dots)$}. Hence, $(a, b){\in}\mathcal{I}$.

$(2)\  b{=}W.$
Suppose $rank(o){=}i$, $rank_a(W){=}j$, then $i{\neq} j$. Since $\issue{o}$ can only block process $i$ and $W$ can only make process $j$ active, $ b{\in} enabled( a(S))$ and $ a{\in} enabled(b(S))$.
In addition, {\small$a(b(S)){=}b(a(S)){=}$\\$(\dots,s_i[update(\mathcal{F},o),Stat{:=}Stat',\mathcal{B}.push(o)],{\dots},s_j[update(\mathcal{F},$\\$\texttt{Wait}),\mathcal{B}.pull(\texttt{Wait})],{\dots})$}. Hence, $(a, b){\in}\mathcal{I}$.

$(3)\  b {=}SR{\lor}b{=}SR^*.$
Suppose $b{=}(p,q)$, $rank(o){=}i$, $rank(p){=}j_1$, and $rank(q){=}j_2$. For $i{\neq} j_1{\land} i{\neq} j_2$: $\issue{o}$ only updates the state $s_i$, and $ready(p,s_{j_1})$ and $ready(q,s_{j_2})$ will not be affected. On the other hand, $ b$ cannot make process $i$ blocked. Hence $b{\in} enabled(a(S))$ and $a{\in} enabled( b(S))$. Then,
{\small$a(b(S)){=} b(a(S)){=}({\dots},s_i[update(\mathcal{F}{,}o),Stat$\\${:=}Stat',\mathcal{B}.push(o)],\dots,s_{j_1}[update(\mathcal{F},p),\mathcal{B}.pull(p)],{\dots},$\\
$ s_{j_2}[update(\mathcal{F},q),\mathcal{B}.pull(q)],{\dots})$}. For
$i{=}j_1$: since $\issue{o}$ and $ b$ are co-enabled, $index(o){>}index(p)$ and $p$ is non-blocking. Due to the condition $index(o){>}index(p)$, $\issue{o}$ has no effect on $ready(p,s_{j_1})$. On the other hand, since $p$ is non-blocking, $b$ cannot make process i blocked. Hence $ b{\in} enabled(a(S))$ and $a{\in} enabled(b(S))$. Additionally, since $p$ is non-blocking, $a(b(S)){=}b( a(S)){=}(\dots,s_i[update(\mathcal{F},o),Stat$\\${:=}Stat', \mathcal{B}.pull(p),\mathcal{B}.push(o)],\!\dots\!,s_{j_2}[update(\mathcal{F},q),\mathcal{B}.pull(q)],\!\dots)$. For $i{=}j_2$, the proof is similar. Hence $(a, b){\in}\mathcal{I}$.
\end{proof}

\begin{thm}\label{thm:correctness}
$E(S)$ preserves the satisfaction of global reachability properties.
\end{thm}
\begin{proof}
We first prove the $E(S)$ satisfies condition $\mathbf{C1}$ and $\mathbf{C2}$, respectively.

\textbf{C1:} $\forall$ $ a\in E(S)$, if $( a, b)\in \mathcal{D}$, then for every trace
$S_0 \xrightarrow{a_0} S_1\xrightarrow{a_1},\dots,\xrightarrow{a_k} S_k\xrightarrow{ b} S_{i+1}$, there exists $a_i\in E(S)$,
where $0\leq i\leq k$.\quad

\textsf{Case 1:}\quad $E(S)=\enable{S}$, \textbf{C1} holds because $a_0\in E(S)$.

\textsf{Case 2:}\quad $E(S)\neq \enable{S}$. In this case, $E(S)$ contains only one element, and can be $\issue{o}$, $B$, $W$, or $SR$.
Assume \textbf{C1} does not hold, \ie, $a_i\neq a$.
According to Theorem~\ref{thm:independence}, $( a, a_i)\in\mathcal{I}$ and $ a\in enabled(S_i)$. Because $ a, b\in enabled(S_k)$, $( a, b)\in\mathcal{I}$,
which conflicts with the premise that $( a, b)\in \mathcal{D}$. Hence \textbf{C1} holds.

\textbf{C2}: Since there is no cycle in the labeled transition system of MPI programs, \textbf{C2} holds.

Since $E(S)\subseteq \enable{S}$, $M'$ is a sub model of $M$. Assume $E(S)$ does not preserve the satisfaction of global reachability property $\varphi$. According to Definition \ref{def-property}, if $\varphi$ is $\mathbf{EF}\ \gamma$, \ie, there exists a state $S$ in $M$ that satisfies $\gamma$, but no state in $M'$ satisfying $\gamma$. Since $E(S)$ satisfies \textbf{C1} and \textbf{C2}, suppose the execution trace to $S$ is $T$, \ie, $\mathsf{r}(T) = S$, then there must exist an equivalent execution trace $T'\in M'$. Obviously, commuting independent actions cannot change the result state, so $\mathsf{r}(T')=S$, which conflicts with the assumption. If $\varphi$ is $\neg \mathbf{EF}\ \gamma$, \ie, each state in $M$ does not satisfy $\gamma$, but there exists a state in $M'$ satisfying $\gamma$, which conflicts with $M'$ is a sub model of $M$. Hence the theorem holds.

\end{proof}


\subsection{Support of Master-Slave Pattern}\label{master-slave-support}

In real-world MPI programs, communications may depend on message contents, which makes the behavior more complicated. For example, suppose a message content is stored in variable $x$, followed by ``$\textbf{if}\ (x > 10)\ \verb"Send(0)"; \textbf{else}\ \verb"Send(1)"$". The general way to handle the situation where communications depend on the message contents is to make the contents symbolic. \mpise make $x$ to be symbolic once it detects that there exist communications depending on $x$, so that \mpise\ does not miss a branch. Master-slave pattern is a representative situation and has been widely adopted to achieve a dynamic load balancing~\cite{gropp1999using}. The verification of the parallel programs employing dynamic load balancing is a common challenge. In the pattern, the master process is in charge of management, \ie, dispatching jobs to slave processes and collecting results. On the other hand, the slave processes repeatedly work on the jobs and send results back to the master process until receiving the message of termination. Figure~\ref{fig:example_master_slave} is an example program.
Following shows an example program.

\begin{figure}[!htb]
\begin{center}
{\footnotesize\begin{tabular}{l|l|l}
\hline
  $P_0$          & $P_1$            &     $P_2$ \\
  \hline
  \verb"Send(1);" & \textbf{while}(true) \textbf{do} \verb"{" & \textbf{while}(true) \textbf{do} \verb"{"\\
  \verb"Send(2);" &\verb"Recv(0);"  &\verb"Recv(0);" \\
  \textbf{while} (...) \textbf{do} \verb"{" &\textbf{if} (...) &\textbf{if} (...) \\
  \verb"Recv(*,r);" & \ \ \textsf{break}\verb";"& \ \ \textsf{break}\verb";"\\
  \verb"Send(r.src)" & \verb"Send(0);"&\verb"Send(0);"\\
  \verb"}" &\verb"}"&\verb"}"\\
  \verb"..."&&\\
  \hline
\end{tabular}}
\end{center}
\caption{An example program using master-slave pattern.}
\label{fig:example_master_slave}
\end{figure}

$P_0$ is the master process, and the remaining processes are slaves. $P_0$ first dispatches one job to each slave. Then, $P_0$ will iteratively receive a job result from any slave (\verb"Recv(*,r)") and dispatch a new job to the slave process whose result is just received, \ie, \verb"Send(r.src)", where \verb"r" is the status of the receive and \verb"r.src" denotes the process identifier of the received message. Each slave process iteratively receives a job (\verb"Recv(0)"), completes the job (omitted for brevity) and sends the job result to $P_0$ (\verb"Send(0)"). The \textbf{if} condition becomes true when the received message is for termination.
The total number of jobs is controlled by the \textbf{while} loop in $P_0$. After all the jobs have been completed, $P_0$ will notify all the slave processes to exit.
The communication behavior of master-slave pattern is highly dynamic, \ie, the destination of the job send operation depends on the matching of the wildcard receive in the master process. In principle, if there are $n$ dynamically dispatched jobs for $m$ slaves, there exists $m^n$ cases of dispatching jobs. \mpise\ supports master-slave pattern as follows.

\textbf{Recognition.} We recognize the master-slave pattern automatically during symbolic execution. More precisely, to recognize the master process, we check whether the process identifier of a received message of a wildcard receive ($\verb"Recv(*,r)"$) is used as the destination of a send operation, \ie, $\verb"Send(r.src)"$. We call such a wildcard receive \texttt{schedule receive}.
Then, we locate the corresponding slave processes \wrt the matchings of the \texttt{schedule receive}, \ie, a process is a slave if its message can match the \texttt{schedule receive}.

\textbf{Modeling the master process.} We allocate a global variable $label$ for each \texttt{schedule receive} to decide the destination process of the next job. In addition to the operations for modeling
normal wildcard receive operations, \ie, using external choice of the matched channel readings, we write the read value from the matched channel to $label$. For example, suppose $Chan_1$ and $Chan_2$ are the matched channels of a \texttt{schedule receive}, we model it by $Chan_1? label\rightarrow \verb"Skip"\square Chan_2 ?label\rightarrow \verb"Skip"$. Considering that a send operation in slave processes is modeled by writing the slave's process identifier to the channel, we can use the value of $label$ to decide the destination of the next job.

\textbf{Modeling the slave process.} We use \emph{recursive} CSP process to model a slave's dynamic feature, \ie, repeatedly receiving a job and sending the result back until receiving the termination message. To model the job receive operation in a slave process,
we use a guard expression $[label==i]$ before the channel reading of the receive operation, where $label$ is the global variable of the corresponding \texttt{schedule receive} and $i$ is the slave process's identifier. Notably, the guard expression will disable the channel reading until the inside condition becomes true, indicating that
the slave process cannot receive a new job unless its result has just been received by the master.

We need to refine the algorithms of symbolic execution and CSP modeling to support master-slave pattern. We have already implemented the refinement in MPI-SV. The
support of master-slave pattern demonstrates that \mpise\ outperforms the single path reasoning work~\cite{forejt2014precise,DBLP:conf/nfm/HuangM15}.

\subsection{Proof of CSP Modeling's Soundness and Completeness}\label{prove_path_correctness}

\noindent\textbf{Theorem 4.1} $\mathcal{F}(\emph{\textsf{CSP}}_{static}) =  \mathcal{F}(\emph{\textsf{CSP}}_{ideal})$.

\begin{proof}

We first prove $\mathcal{T}(\textsf{CSP}_{static}) =  \mathcal{T}(\textsf{CSP}_{ideal})$, based on which we can prove $\mathcal{F}(\emph{\textsf{CSP}}_{static}) =  \mathcal{F}(\emph{\textsf{CSP}}_{ideal})$.

First, we prove $\mathcal{T}(\textsf{CSP}_{static}) {\subseteq}  \mathcal{T}(\textsf{CSP}_{ideal})$ by contradiction. Suppose there exists a trace $t {=} \langle e_1, ...,  e_n\rangle$ such that $ t {\in} \mathcal{T}(\textsf{CSP}_{static})$ but $t {\notin} \mathcal{T}(\textsf{CSP}_{ideal})$. The only difference between $\textsf{CSP}_{static}$ and $\textsf{CSP}_{ideal}$ is that $\textsf{CSP}_{static}$ introduces more channel read operations during the modeling of receive operations. Hence, there must exist a read operation of an extra channel in $t$. Suppose the first extra read is $e_k {=} c_e?x$, where $1{\le} k {\le} n$. Therefore, $c_e$ \emph{cannot} be read in $\textsf{CSP}_{ideal}$ when the matching of the corresponding receive operation starts, but $c_e$ is not empty at $e_k$ in $\textsf{CSP}_{static}$. Despite of the size of $c_e$, there must exist a write operation $c_e!y$ in $\langle e_1, ...,  e_{k-1}\rangle$. Because $\langle e_1, ...,  e_{k-1}\rangle$ is also a valid trace in $\textsf{CSP}_{ideal}$, it means $c_e$ is not empty in $\textsf{CSP}_{ideal}$ at $e_k$, which contradicts with the assumption that $c_e$ cannot be read in $\textsf{CSP}_{ideal}$. Hence, $\mathcal{T}(\textsf{CSP}_{static}) \subseteq  \mathcal{T}(\textsf{CSP}_{ideal})$ holds.

Then, we prove $\mathcal{T}(\textsf{CSP}_{ideal}) {\subseteq}  \mathcal{T}(\textsf{CSP}_{static})$ also by contradiction. Suppose there exists a trace $t {=} \langle e_1, ...,  e_m\rangle$ that $ t {\in} \mathcal{T}(\textsf{CSP}_{ideal})$ but $t {\notin} \mathcal{T}(\textsf{CSP}_{static})$. Because $\textsf{SMO}(op_j, S)$ is a superset of the precise matching set of $op_j$, $t$ cannot be a terminated trace. So, $\textsf{CSP}_{ideal}$ blocks at $e_m$. Because $t \notin \mathcal{T}(\textsf{CSP}_{static})$, there must exist a channel read operation $c_m?x$ that is enabled at $e_m$ in $\textsf{CSP}_{static}$, \emph{i.e.}, $c_m$ is not empty. Hence, there must exist a write operation $c_m!y$ in $\langle e_1, ...,  e_{m-1}\rangle$. Because $\langle e_1, ...,  e_{m-1}\rangle$ is valid in both of $\textsf{CSP}_{static}$ and $\textsf{CSP}_{ideal}$, $c_m?x$ is also enabled at $e_m$ in $\textsf{CSP}_{ideal}$, which contradicts with the assumption that $\textsf{CSP}_{ideal}$ blocks. Hence, we can have {\small$\mathcal{T}(\textsf{CSP}_{ideal}) {\subseteq}  \mathcal{T}(\textsf{CSP}_{static})$}, and 
$\mathcal{T}(\mathsf{CSP}_{static}) {=}  \mathcal{T}(\mathsf{CSP}_{ideal})$ holds.

Next, we can prove $\mathcal{F}(\emph{\textsf{CSP}}_{static}) =  \mathcal{F}(\emph{\textsf{CSP}}_{ideal})$ in a similar way. 
Suppose there exists $(s, X)$ in $\textsf{CSP}_{static}$ but $(s, X) \notin \textsf{CSP}_{ideal}$. It means there exists an event \emph{e} in $X$ that is refused by $\textsf{CSP}_{static}$ at $s$, but enabled by $\textsf{CSP}_{ideal}$ at $s$. Because there is no internal choice in the CSP models, we have $s {\cdot} \langle e\rangle {\notin} \mathcal{T}(\textsf{CSP}_{static})$~\citep{roscoe1998theory} and  $s {\cdot} \langle e\rangle {\in} \mathcal{T}(\textsf{CSP}_{ideal})$, which conflicts with {\small$\mathcal{T}(\mathsf{CSP}_{static}) {=}  \mathcal{T}(\mathsf{CSP}_{ideal})$}. The contradiction of the case in which $(s, X) {\notin} \textsf{CSP}_{static}$ but $(s, X) {\in}$\\$ \textsf{CSP}_{ideal}$ can be proved similarly. 

Finally, $\mathcal{F}(\emph{\textsf{CSP}}_{static}) =  \mathcal{F}(\emph{\textsf{CSP}}_{ideal})$ is proved.
\end{proof}

\noindent\textbf{Theorem 4.2} 
$\emph{\textsf{CSP}}_{static}$ \emph{is consistent with the MPI semantics}.

\begin{proof}
If the global state of generating \textsf{CSP}$_{static}$ is $S_c$, then we can get an MPI program $\mathcal{MP}_p$ from the sequence set $\mathsf{Seq}(S_c)$, where each process $\mathsf{Proc}_i$ of $\mathcal{MP}_p$ is the sequential composition of the operations in $Seq_i$. Suppose the LTS model of $\mathcal{MP}_p$ is $M_p$, and the LTS after hiding all the $\issue{o}$ actions in $M_p$ is $\hat{M_p}$. Then, $\mathsf{CSP}_{static}$ is consistent with the MPI semantics iff
$\{(\mathsf{M}_t(s),\mathsf{M}_s(X))  \mid (s, X) \in$\\$\mathcal{F}(\textsf{CSP}_{static})\}$ is equal to $\{(T, X) {\mid} T {\in} \hat{M_p} {\wedge} X {\subseteq} \mathsf{M}_s(\Sigma) {\setminus} \enable{\mathsf{r}(T)}\}$, where $\Sigma$ is the event set of $\textsf{CSP}_{static}$, $\mathsf{M}_t(s)$ and $\mathsf{M}_s(X)$ maps the events in the sequence $t$ and the set $X$ to the corresponding actions in MPI semantics, respectively.
This can be shown by proving that Algorithms~\ref{alg:Modeling} with a precise $\mathsf{SMO}$ ensures all the \emph{completes-before} relations of MPI semantics (\cf semantic rules in Figure \ref{rules}). The relations between send operations and those between receive operations (including conditional completes-before relation) are ensured by $\textsf{Refine}(P, S)$. The communications of $send$ and $recv$ operations are modeled by CSP channel operations and process compositions. The requirements of  \verb"Wait" and \verb"Barrier" operations are modeled by the process compositions defined in Algorithm~\ref{alg:Modeling}. Hence, we can conclude that $\textsf{CSP}_{ideal}$ is consistent with the MPI semantics. Then, by Theorem \ref{thm:failure}, we can prove $\textsf{CSP}_{static}$ is consistent with the MPI semantics.
\end{proof}

\fi

\end{document}